\documentclass[a4paper,reqno]{amsart}
\usepackage{amsmath,amssymb,mathrsfs}
\usepackage{txfonts,bm,pifont}
\usepackage{cite}
\usepackage[dvipdfm,bookmarks=true,bookmarksnumbered=true,colorlinks=true]{hyperref}
\usepackage{comment}
\usepackage[dvipdfmx]{graphicx}
\PassOptionsToPackage{svgnames}{xcolor}
\usepackage{tikz}
\usetikzlibrary{calc}
\usetikzlibrary{decorations.markings,decorations.pathmorphing}

\newtheorem{theorem}{Theorem}[section]

\newtheorem{lemma}[theorem]{Lemma}
\newtheorem{remark}[theorem]{Remark}
\newtheorem{definition}[theorem]{Definition}
\numberwithin{equation}{section}
\newcommand{\sn}[1]{{\rm sn}\left(#1\right)}
\newcommand{\cn}[1]{{\rm cn}\left(#1\right)}
\newcommand{\dn}[1]{{\rm dn}\left(#1\right)}
\newcommand{\cd}[1]{{\rm cd}\left(#1\right)}

\newcommand{\pP}[2]{P_{#1}\left(#2\right)}
\newcommand{\pQ}[2]{Q_{#1}\left(#2\right)}
\newcommand{\gG}[1]{G_{#1}\,}
\newcommand{\bfrac}[2]{\left(\frac{#1}{#2}\right)}
\newcommand{\bcfrac}[2]{\left(\cfrac{#1}{#2}\right)}
\newcommand{\lrangle}[1]{\langle#1\rangle}
\newcommand{\ox}{\overline{x}}
\newcommand{\oy}{\overline{y}}
\newcommand{\tx}{\tilde{x}}

\newcommand{\ty}{\tilde{y}}

\newcommand{\io}{\iota}
\newcommand{\bbP}{\mathbb{P}}
\newcommand{\bbZ}{\mathbb{Z}}
\newcommand{\bbR}{\mathbb{R}}
\newcommand{\PicX}{{\rm Pic}(X)}
\newcommand{\WEe}{W\big(E_8^{(1)}\big)}
\newcommand{\tWEe}{\widetilde{W}\big(E_8^{(1)}\big)}
\newcommand{\aaa}[1]{a_{#1}}
\newcommand{\cc}[1]{c_{#1}}
\newcommand{\EE}[1]{E_{#1}}

\newcommand{\RJ}[1]{R_{J,#1}}
\newcommand{\gae}{\gamma_{\rm e}}
\newcommand{\gao}{\gamma_{\rm o}}
\newcommand{\sz}{{\rm sz}}
\newcommand{\cz}{{\rm cz}}
\newcommand{\dz}{{\rm dz}}
\newcommand{\hsz}{{\rm \widehat{sz}}}
\newcommand{\hcz}{{\rm \widehat{cz}}}
\newcommand{\hdz}{{\rm \widehat{dz}}}
\newcommand{\sge}{{\rm sg}_{\rm e}}
\newcommand{\cge}{{\rm cg}_{\rm e}}
\newcommand{\dge}{{\rm dg}_{\rm e}}
\newcommand{\sgo}{{\rm sg}_{\rm o}}
\newcommand{\cgo}{{\rm cg}_{\rm o}}
\newcommand{\dgo}{{\rm dg}_{\rm o}}
\newcommand{\al}{\alpha}

\newcommand{\de}{\delta}
\newcommand{\ga}{\gamma}
\newcommand{\si}{\sigma}

\newcommand{\la}{\lambda}

\newcommand{\ep}{\bm{\epsilon}}
\newcommand{\ka}{\kappa}
\newcommand{\tW}{\widetilde{W}}
\newcommand{\oc}[1]{{#1}^{\vee}}
\newcommand{\utilde}[1]{\vrule depth 0pt width 0pt%
{\smash{{\mathop{#1}\limits_{\displaystyle\tilde{}}}}}}
\newcommand{\wutilde}[1]{\vrule depth 0pt width 0pt%
{\raise0.8pt\hbox{$\smash{{\mathop{#1} \limits_{\displaystyle\widetilde{}}}}$}}}
\makeatletter
\long\def\@makecaption#1#2{
 \vskip 10pt
 \setbox\@tempboxa\hbox{#1. #2}
 \ifdim \wd\@tempboxa >\hsize #1. #2\par \else \hbox
to\hsize{\hfil\box\@tempboxa\hfil}
 \fi}
\newcommand\Pfour{\textrm{P}_{\textrm{IV}}}

\newcommand{\iii}{{\rm i}}
\DeclareFontFamily{U}{mathx}{\hyphenchar\font45}
\DeclareFontShape{U}{mathx}{m}{n}{
      <5> <6> <7> <8> <9> <10>
      <10.95> <12> <14.4> <17.28> <20.74> <24.88>
      mathx10
      }{}
\DeclareSymbolFont{mathx}{U}{mathx}{m}{n}
\DeclareFontSubstitution{U}{mathx}{m}{n}
\DeclareMathAccent{\widecheck}{0}{mathx}{"71}
\DeclareMathAccent{\wideparen}{0}{mathx}{"75}

\usepackage{geometry,multicol}
\newdimen\stockheight
\newdimen\stockwidth
\stockwidth\paperwidth
\begin{document}
\title{A review of elliptic difference Painlev\'e equations}
\author{Nalini Joshi}
\address{School of Mathematics and Statistics, The University of Sydney, New South Wales 2006, Australia.}
\email{nalini.joshi@sydney.edu.au}
\author{Nobutaka Nakazono}
\address{Department of Physics and Mathematics, Aoyama Gakuin University, Sagamihara, Kanagawa 252-5258, Japan.}
\email{nobua.n1222@gmail.com}

\begin{abstract}
  \noindent
Discrete Painlev\'e equations are nonlinear, nonautonomous difference equations of second-order. They have coefficients that are explicit functions of the independent variable $n$ and there are three different types of equations according to whether the coefficient functions are linear, exponential or elliptic functions of $n$. In this paper, we focus on the elliptic type and give a review of the construction of such equations on the $E_8$ lattice. The first such construction was given by Sakai \cite{SakaiH2001:MR1882403}. We focus on recent developments giving rise to more examples of elliptic discrete Painlev\'e equations. 
\end{abstract}

\maketitle
\setcounter{tocdepth}{1}
\section{Introduction}
\label{section:introduction}
Discrete Painlev\'e equations are nonlinear integrable ordinary difference equations of second order. They have a long history (see \S\ref{s:background}), but it is only in the past two decades that striking developments have led to their recognition as one of the most important classes of equations in the theory of integrable systems \cite{GR2004:MR2087743,KNY2017:MR3609039,HJN2016:MR3587455}. 

Almost all of the currently known collection of discrete Painlev\'e equations were derived by Grammaticos, Ramani and collaborators \cite{GR2004:MR2087743}, and recognized as having fundamental properties that parallel those of the Painlev\'e equations, such as Lax pairs, B\"acklund transformations and special solutions. 
Sakai \cite{SakaiH2001:MR1882403} unified these discrete Painlev\'e equations and also discovered a new equation whose coefficients are iterated on elliptic curves.  Sakai's equation is an elliptic difference Painlev\'e equation.

Sakai's unification is based on a deep geometric theory shared by all the discrete Painlev\'e equations, first described by Okamoto \cite{OkamotoK1979:MR614694} for the classical Painlev\'e equations (see \S\ref{s:background}). The fundamental unifying property is based on the fact that the initial-value (or phase) space of the Painlev\'e equations can be compactified and regularized by a minimum of eight blow-ups on a Hirzebruch surface. This beautiful observation also leads to symmetries of the equation, arising from an isomorphism between the intersection diagram of the resulting space and affine root systems. 

This led Sakai to describe discrete Painlev\'e equations as the result of translations on lattices defined by affine Weyl groups \cite{book_HumphreysJE1992:Reflection}. In particular, Sakai's elliptic difference equation \cite{SakaiH2001:MR1882403} is iterated on the lattice generated by the affine exceptional Lie group $E_8^{(1)}$ (see also \cite{MSY2003:MR1958273,ORG2001:MR1877472}). 
More recently, other elliptic difference equations of Painlev\'e type have been discovered \cite{RCG2009:MR2525848,AHJN2016:MR3509963,JN2017:MR3673467,CDT2017:MR3708091} through other approaches. This review concentrates on describing the construction of such elliptic difference Painlev\'e equations by using Sakai's geometric way.

To describe the construction that underlies and explains all of these examples, we rely on the following mathematical description. Fix a point in the $E_8^{(1)}$ lattice \cite{book_CSBBLNOPQV2013:Sphere}. 
Then there are 240 nearest neighbors of this point in the lattice, lying at a distance whose squared length is equal to 2. We refer to the 120 vectors between the initial fixed point and its possible nearest neighbors as nearest-neighbor-connecting vectors (NVs). Similarly, there are 2160 next-nearest neighbors, lying at a distance whose squared length is 4. The 1080 vectors between the fixed point and such next-nearest neighbors will be referred to as next-nearest-neighbor-connecting vectors (NNVs). 
Sakai's elliptic difference equation is constructed in terms of translations expressed in terms of NVs. 
However, more recently deduced examples are obtained from NNVs.

The example that led us to this key observation is the following elliptic difference Painlev\'e equation, originally found by Ramani, Carstea and Grammaticos\cite{RCG2009:MR2525848}:
\begin{equation}\label{eqn:RCGeqn}
\begin{cases}
 y_{n+1}\Big(k^2(\cge^2-\cz_n^2)\cz_n\dz_n x_n^2 y_n-(1-k^2 \sz_n^4)\cge\dge x_n\\
 \hspace{3em}+(1-k^2\sge^2\sz_n^2)\cz_n\dz_n y_n\Big)\\
 \hspace{1.2em}=(1-k^2\sz_n^4)\cge\dge x_n y_n-(\cge^2-\cz_n^2)\cz_n\dz_n\\
 \hspace{3em}-(1-k^2\sge^2 \sz_n^2)\cz_n\dz_n x_n^2,\\[1em]
 x_{n+1}\Big(k^2(\cgo^2-\hcz_n^2)\hcz_n\hdz_n y_{n+1}^2x_n-(1-k^2 \hsz_n^4)\cgo\dgo y_{n+1}\\
 \hspace{3em}+(1-k^2\sgo^2 \hsz_n^2)\hcz_n\hdz_n x_n\Big)\\
 \hspace{1.2em}=(1-k^2\hsz_n^4)\cgo\dgo y_{n+1} x_n-(\cgo^2-\hcz_n^2)\hcz_n\hdz_n\\
 \hspace{3em}-(1-k^2\sgo^2 \hsz_n^2)\hcz_n\hdz_n y_{n+1}^2,
\end{cases}
\end{equation}
where
\begin{subequations}
\begin{align}
 &\sz_n=\sn{z_n},\quad
 \hsz_n=\sn{z_n+\gae+\gao},\quad
 \sge=\sn{\gae},\\
 &\sgo=\sn{\gao},\quad
 \cz_n=\cn{z_n},\quad
 \hcz_n=\cn{z_n+\gae+\gao},\\
 &\cge=\cn{\gae},\quad
 \cgo=\cn{\gao},\quad
 \dz_n=\dn{z_n},\\
 &\hdz_n=\dn{z_n+\gae+\gao},\quad
 \dge=\dn{\gae},\quad
 \dgo=\dn{\gao},
\end{align}
\end{subequations}
and $z_n=z_0+2(\gae+\gao)n$.
We will refer to this equation as the RCG equation.
Here, ${\rm sn}$, ${\rm cn}$ and ${\rm dn}$ are Jacobi elliptic functions and $k$ is the modulus of the elliptic sine. 
For more information about Jacobian elliptic functions and notations, see \cite[Chapter 22]{NIST:DLMF} and \cite{book_WW1996:course}.

It turns out that the RCG equation \eqref{eqn:RCGeqn} is a {\em projective reduction} of an NNV, i.e., the iterative step is not a translation on the $E_8^{(1)}$ lattice, but its square is a translation corresponding to a NNV. In general, we can derive various discrete Painlev\'e equations from elements of infinite order in the affine Weyl group that are not necessarily translations by taking a projection on a certain subspace of the parameters. Kajiwara {\em et al} studied such procedures \cite{KNT2011:MR2773334,KN2015:MR3340349} to obtain non-translation type discrete Painlev\'e equations  and named these \lq\lq projective reductions\rq\rq . The RCG equation \eqref{eqn:RCGeqn} is the first known case of an elliptic difference Painlev\'e equation obtained from such a reduction.

We constructed a discrete Painlev\'e equation, which has the RCG equation as a projective reduction, by using NNVs on the $E_8^{(1)}$ lattice in \cite{JN2017:MR3673467}.  Most of discrete Painlev\'e equations admit the special solutions expressible in terms of solutions of linear equations 
when some of the parameters take special values (see, for example, \cite{KNY2017:MR3609039} and references therein). 
It is known that projectively-reduced equations have different type of such solutions from translation-type discrete Painlev\'e equations on the same lattice\cite{KNT2011:MR2773334,KN2015:MR3340349}.
In this paper, we will also show the special solutions of Equation \eqref{eqn:RCGeqn}, which is quite different from those of a translation-type elliptic difference Painlev\'e equation reported in \cite{KNY2017:MR3609039,KMNOY2003:MR1984002}.
\subsection{Background}\label{s:background}
Shohat studied polynomials $\Phi_n(x)$ indexed by degree $n\in\mathbb N$, defined in an interval $(-\infty, \infty)$, with a weight function $p(x)=\exp(-x^4/4)$ such that
\begin{equation}\label{eq:ShohatOrthogRel}
  \int_{-\infty}^{\infty}p(x)\Phi_m(x)\,\Phi_n(x)dx=0,\quad (m\not=n,~ m,n\in\mathbb N).
\end{equation}
Shohat obtained the 3-term recurrence relation (see Equation (39) of \cite{s:39})
\begin{equation}\label{eq:Shohat3term}
  \Phi_n(x)-(x-c_n)\,\Phi_{n-1}(x)+\lambda_n\,\Phi_{n-2}(x)=0,\quad n\ge 2,
\end{equation}
where $\Phi_0(x)\equiv 1$, $\Phi_1(x)=x-c_1$, where $c_1$ is independent of $x$, and deduced the following difference equation for $\lambda_n$:
\begin{equation}\label{eq:ShohatdP1}
  \lambda_{n+2}\bigl(\lambda_{n+1}+\lambda_{n+2}+\lambda_{n+3}\bigr)=n+1.
\end{equation}

We now know that this equation is intimately related to one of the six classical Painlev\'e equations, universal classes of second-order ordinary differential equations (ODEs) studied by Painlev\'e \cite{p:02}, Fuchs \cite{f:05} and  Gambier \cite{gambier1910equations}. Fokas \emph{et al} \cite{fokas1991discrete} showed that the solutions of Equation \eqref{eq:ShohatdP1} are solutions of the fourth Painlev\'e equation:
\begin{equation*}
 \Pfour:\quad 
 w''=\frac{{w'}^2}{2w}+\frac{3 w^3}{2}+4tw^2+2(t^2-\al)w+\frac{\beta}{w}. 
\end{equation*}
Actually, solutions of $\Pfour$: $w=w_n$, $n\in\bbZ$, satisfy a more general version of Shohat's equation given by 
\begin{equation}\label{eq:dP1full}
  w_n\,\bigl(w_{n+1}+w_n+w_{n-1}\bigr)=a\,n+b+c\,(-1)^n+d\,w_n,
\end{equation}
where $a$, $b$, $c$, $d$ are constants (see \cite{FGR93}). This equation is an integrable equation in its own right, with fundamental properties such as a Lax pair \cite{cj:99}.

Equation \eqref{eq:dP1full} is one of many equations now known as discrete Painlev\'e equations. 
In general, there exist three types of discrete Painlev\'e equations. 
They are distinguished by the types of function $t_n$ appearing in the coefficient of each equation.
\begin{enumerate}
\item
If there exists $k\in\bbZ_{>0}$ such that $t_{n+k}-t_n$ is a constant, 
then the equation is said to be of {\em additive-type}.
\item
If there exists $k\in\bbZ_{>0}$ such that $t_{n+k}/t_n$ is a constant, denoted $q$ $(\not=0, 1)$, then the equation is said to be of {\em multiplicative-type} or {\em $q$-difference-type}.
\item
If $t_n$ can be expressed by the elliptic function of $n$, then the equation is said to be of {\em elliptic-type}.
\end{enumerate} 
We list a few discrete Painlev\'e equations here.
\begin{align}
 \text{d-P$_{\rm II}$}\text{\cite{PS1990:Unitary-matrix}}:~
 &X_{n+1}+X_{n-1}=\cfrac{t_nX_n+a}{1-{X_n}^2}\label{eqn:dp2},\\
 \text{$q$-P$_{\rm III}$}\text{\cite{RGH1991:MR1125951}}:~
 &X_{n+1}X_{n-1}=\cfrac{(X-at_n)(X-a^{-1}t_n)}{(X-b)(X-b^{-1})},\label{eqn:qp3}\\
 \text{d-P$_{\rm IV}$}\text{\cite{RGH1991:MR1125951}}:~
 &(X_{n+1}+X_n)(X_n+X_{n-1})
 =\cfrac{({X_n}^2-a^2)({X_n}^2-b^2)}{(X_n-t_n)^2-c^2},\label{eqn:dp4}\\
 \text{$q$-P$_{\rm V}$}\text{\cite{RGH1991:MR1125951}}:~
 &(X_{n+1}X_n-1)(X_nX_{n-1}-1)\notag\\
 &={t_n}^2\cfrac{(X_n-a)(X_n-a^{-1})(X_n-b)(X_n-b^{-1})}{(X_n-ct_n)(X_n-c^{-1}t_n)},\label{eqn:qp5}
\end{align}
where $a$, $b$ and $c$ are constants. 
Here, $t_{n+1}-t_n$ is a constant for Equations \eqref{eqn:dp2} and \eqref{eqn:dp4} 
and $t_{n+1}/t_n$ is a constant for Equations \eqref{eqn:qp3} and \eqref{eqn:qp5},
that is, Equations \eqref{eqn:dp2} and \eqref{eqn:dp4} are additive-type,
while Equations \eqref{eqn:qp3} and \eqref{eqn:qp5} are multiplicative-type.
Note that the notation for each equation originates from their discrete types and continuum limits.
Moreover, an example of elliptic-type is given by Equation \eqref{eqn:RCGeqn}.

Okamoto \cite{OkamotoK1979:MR614694} described a geometric framework for studying the Painlev\'e equations. He showed that the initial-value (or phase) space of the Painlev\'e equations, which is a foliated vector bundle \cite{milnor1970foliations},  can be compactified and regularised by a minimum of eight blow-ups on a Hirzebruch surface (or nine in $\mathbb P^2$).

This geometric theory also leads to a description of their symmetry groups, described in terms of affine Weyl groups \cite{n:04}. Such symmetries lead to transformations of the Painlev\'e equations called \emph{B\"acklund transformations}. 

Sakai's geometric description of discrete Painlev\'e equations, based on types of space of initial values,
is well known\cite{SakaiH2001:MR1882403}.
This picture relies on compactifying and regularizing space of initial values.
The spaces of initial values are constructed by  blow up of $\bbP^1\times\bbP^1$ at  base points (see \S \ref{section:General setting}) 
and are classified into 22 types according to the configuration of the base points as follows:
\begin{center}
\begin{tabular}{|l|l|}
\hline
Discrete type&Type of surface\\
\hline
Elliptic&$A_0^{(1)}$\rule[-.5em]{0em}{1.6em}\\
\hline
Multiplicative&$A_0^{(1)\ast}$, $A_1^{(1)}$, $A_2^{(1)}$, $A_3^{(1)}$, \dots, $A_8^{(1)}$, $A_7^{(1)'}$\rule[-.5em]{0em}{1.6em}\\
\hline
Additive&$A_0^{(1)\ast\ast}$, $A_1^{(1)\ast}$, $A_2^{(1)\ast}$, $D_4^{(1)}$, \dots, $D_8^{(1)}$, $E_6^{(1)}$, $E_7^{(1)}$, $E_8^{(1)}$\rule[-.5em]{0em}{1.6em}\\
\hline
\end{tabular}
\end{center}
In each case, the root system characterizing the surface forms a subgroup of the 10-dimensional Picard lattice. 
The symmetry group of each equation, formed by Cremona isometries, arises from the orthogonal complement of this root system inside the Picard lattice.
\subsection{Periodic reduction of the Q4-equation}
\label{subsection:periodic_reduction}
In this section, we recall how to obtain Equation \eqref{eqn:RCGeqn} 
from the lattice Krichever-Novikov system\cite{AdlerVE1998:MR1601866,HietarintaJ2005:MR2217106} (or, Q4 in the terminology of  Adler-Bobenko-Suris\cite{ABS2003:MR1962121}):
\begin{align}\label{eqn:Q4eqn}
 &\sn{\al_l}(u_{l,m}u_{l+1,m}+u_{l,m+1}u_{l+1,m+1})-\sn{\beta_m}(u_{l,m}u_{l,m+1}+u_{l+1,m}u_{l+1,m+1})\notag\\
 &~-\sn{\al_l-\beta_m}\big(u_{l+1,m}u_{l,m+1}+u_{l,m}u_{l+1,m+1}\big)\notag\\
 &~+\sn{\al_l}\sn{\beta_m}\sn{\al_l-\beta_m}\big(1+k^2u_{l,m}u_{l+1,m}u_{l,m+1}u_{l+1,m+1}\big)=0,
\end{align}
where $\al_l$ and $\beta_m$ are parameters, $u_{l,m}$ is the dependent variable, $l$ and $m$ are independent variables (often taken to be integer) and $k$ is the modulus of the elliptic function $\rm sn$. Taking a periodic reduction $u_{l+1,m-1}=u_{l,m}$ of Equation \eqref{eqn:Q4eqn}, and identifying $n=l+m$, leads to an autonomous second order ordinary difference equation \cite{JGTR2006:MR2271126}:	
\begin{align}\label{eqn:aut_Q4eqn}
 &\big(\sn{\al_0}-\sn{\beta_0}\big)u_n(u_{n+1}+u_{n-1})-\sn{\al_0-\beta_0}(u_{n+1}u_{n-1}+{u_n}^2)\notag\\
 &+\sn{\al_0}\sn{\beta_0}\sn{\al_0-\beta_0}(1+k^2{u_n}^2u_{n+1}u_{n-1})=0.
\end{align}
Ramani {\em et al} \cite{RCG2009:MR2525848} deautonomised Equation \eqref{eqn:aut_Q4eqn} by the method of singularity confinement after a change of variables $\al_0\rightarrow \gamma+z$\,, $\beta_0\rightarrow \gamma -z$.
The resulting equation then becomes Equation \eqref{eqn:RCGeqn} with $x_n=u_{2n}$ and $y_n=u_{2n-1}$.
\subsection{Outline of the paper}
\label{subsection:outline}
In \S \ref{section:weight_lattice_E8}, 
we recall the basic definitions of reflection group theory and define translations, for the interested reader.
The initial-value space of the RCG equation \eqref{eqn:RCGeqn} is constructed in \S \ref{section:General setting}, where we also introduce related algebro-geometric concepts. 
In \S \ref{section:cremona},
we construct Cremona isometries on this initial value space, which roughly speaking, are mappings that preserve its geometric structure.
In \S \ref{section:birational},
we give the resulting birational actions on the coordinates and parameters of the initial value space.
We show that by using these birational actions, we arived at the RCG equation. Some explicit special solutions of the RCG equation are described in
\S \ref{section:Special sols RCG equation}. 
In Appendix \ref{appendix:A4} we illustrate the geometric ideas for the case of $A_4^{(1)}$, and
in Appendix \ref{appendix:general_elliptic} we describe the generic known examples of elliptic difference Painlev\'e equations.
\section{$E_8^{(1)}$-lattice}
\label{section:weight_lattice_E8}
To understand how to construct discrete Painlev\'e equations from symmetry groups, we need the theory of finite and affine reflection groups. In this section, we recall the basic definitions of reflection group theory before defining the transformation group $W(E_8^{(1)})$ and describing its translations operators.

Consider two $n$-dimensional real vector spaces $V$ and $V^*$, spanned by the basis sets $\Delta=\bigl\{\al_1, \ldots, \al_n\bigr\}$ and $\Delta^\vee= \bigl\{\al_1^\vee, \ldots, \al_n^\vee\bigr\}$ respectively. The elements of $\Delta$ are called {\em simple roots}, while those of $\Delta^\vee$ are {\em simple coroots}. To define reflections, we use a bilinear pairing given by the entries of an $n\times n$ Cartan matrix $A=(A_{ij})$ (see the definition of Cartan matrices in \cite{bourbaki2007groupes}):
 \begin{equation}
  \langle\oc\al_i,\al_j\rangle=A_{ij},
\end{equation}
for all $i, j\in \{1, \ldots, n\}$. If $V$ and $V^*$ are Euclidean spaces, this bilinear pairing is the usual inner product. 

It is also important to define the fundamental {\em weights} $h_i$, $i=1, \ldots, n$, which are given by
\begin{equation}\label{eq:h}
\langle \oc\al_i,h_j\rangle=\delta_{ij}, \quad (1\leq i, j\leq n).
\end{equation}
Correspondingly, the integer linear combinations (or $\bbZ$-modules)
\begin{equation}
 Q=\sum_{k=1}^n\bbZ\al_k,\quad
 \oc Q=\sum_{k=1}^n\bbZ\oc\al_k,\quad
 P=\sum_{k=1}^n\bbZ h_k,
\end{equation}
are called the {\em root lattice}, {\em coroot lattice} and {\em weight lattice} respectively.
We are now in a position to define reflections. For each $i\in\{1, \ldots, n\}$, the linear operator defined by
\begin{equation}
  s_{\al_i}(\al_j)=\al_j-A_{ji}\al_i,\quad
  s_{\al_i}(\oc\al_j)=\oc\al_j-A_{ji}\oc\al_i
\end{equation}
  for every $j\in \{1, \ldots, n\}$ is a {\em reflection operator}. That is, it has the following properties:
  \begin{enumerate}
      \item $s_{\al_i}(\al_i)=-\al_i$,\quad $s_{\al_i}(\oc\al_i)=-\oc\al_i$.
      \item $s_{\al_i}^2=1$.
      \item $s_{\al_i}.\langle \oc\al_j, \al_k\rangle
      =\langle s_{\al_i}(\oc\al_j), s_{\al_i}(\al_k)\rangle
      =\langle \oc\al_j, \al_k\rangle$.
  \end{enumerate}
 The group $W$ generated by $s_{\al_1}$, $\ldots$, $s_{\al_n}$ is called a Weyl group. The root system of $W$ is defined to be the subset $\Phi$ of $Q$ given by $\Phi=W(\Delta)$. A root system is said to be irreducible if it is not a combination of mutually orthogonal root systems. 
Each irreducible root system $\Phi$ contains a unique root given by
\begin{equation}\label{eq:maxroot}
  \widetilde\al=\sum_i C_i\al_i,
\end{equation}
whose height (i.e., the sum of coefficients in the expansion) is maximal, which is called the \emph{highest root}. 

For our case, $n=8$, and the starting point is the Cartan matrix of type $E_8$ 
\begin{equation}\label{eqn:finite_cartan_E8}
 A=
 \begin{pmatrix}
 2&-1&0&0&0&0&0&0\\
 -1&2&-1&0&0&0&0&0\\
 0&-1&2&-1&0&0&0&-1\\
 0&0&-1&2&-1&0&0&0\\
 0&0&0&-1&2&-1&0&0\\
 0&0&0&0&-1&2&-1&0\\
 0&0&0&0&0&-1&2&0\\
 0&0&-1&0&0&0&0&2
 \end{pmatrix}.
\end{equation}
The astute reader might notice that this is not the standard one given in Bourbaki \cite{bourbaki2007groupes}, but it is equivalent to it (under conjugate transforms of the bases). We make this choice because it corresponds to the identification of simple roots made in Section \ref{section:cremona}. Moreover, the highest root and coroot are given respectively by
\begin{subequations}\label{eqn:highest_root_E8}
\begin{align}
 &\widetilde\al=2\al_1+4\al_2+6\al_3+5\al_4+4\al_5+3\al_6+2\al_7+3\al_8,\\
 &\widetilde\al^\vee=2\oc\al_1+4\oc\al_2+6\oc\al_3+5\oc\al_4+4\oc\al_5+3\oc\al_6+2\oc\al_7+3\oc\al_8.
\end{align}
\end{subequations}

We now expand the root system by defining $\al_0$ by $\al_0+\widetilde\al=0$. The corresponding extension of the coroot system is defined by $\oc\al_0$ and $\oc\de$ (called the {\em null coroot}): $\oc\al_0+\widetilde\al^\vee=\oc\de$.
To construct the corresponding affine Weyl group, we now extend the Cartan matrix $A$ by adding a row and column given respectively by
\begin{equation}\label{eq:extrarowA}
A_{j0}=\langle -\widetilde\al^\vee,\al_j\rangle,\quad
A_{0j}=\langle \al_j^\vee,-\widetilde\al\rangle,
\end{equation}
along with $A_{00}=2$. The extended Cartan matrix and corresponding root systems and groups are now denoted with a superscript containing $(1)$ to denote this extension. 

The extended Cartan matrix of type $E_8^{(1)}$ is given by
\begin{equation}\label{eqn:cartan_E8}
 A^{(1)}=(A_{ij})_{i, j=0}^8
 =\begin{pmatrix}
 2&0&0&0&0&0&0&-1&0\\
 0&2&-1&0&0&0&0&0&0\\
 0&-1&2&-1&0&0&0&0&0\\
 0&0&-1&2&-1&0&0&0&-1\\
 0&0&0&-1&2&-1&0&0&0\\
 0&0&0&0&-1&2&-1&0&0\\
 0&0&0&0&0&-1&2&-1&0\\
 -1&0&0&0&0&0&-1&2&0\\
 0&0&0&-1&0&0&0&0&2
 \end{pmatrix}.
\end{equation}

The reflections have the form 
\begin{equation}\label{eqn:s_weight_E8}
s_i(\lambda)=\lambda-\langle \al^\vee_i, \lambda\rangle \al_i,\quad i=0,\dots,8,~\la\in P.
\end{equation}
In particular, this gives $s_i(h_i)=h_i-\al_i, i=0,\dots,8$.
On the other hand, each root $\al_i$ can be expressed in terms of weights $h_i$ by using the relationship 
\[
\al_i=\sum_{j=0}^8 A_{ij}h_j.
\]
Therefore, we obtain the reflections on the fundamental weights as follows:
\begin{equation}\label{eqn:s_h_E8}
\begin{aligned}
 &s_0(h_0)=-h_0+h_7, &&s_1(h_1)=-h_1+h_2,&\\
 &s_2(h_2)=h_1-h_2+h_3, &&s_3(h_3)=h_2-h_3+h_4+h_8,&\\
 &s_4(h_4)=h_3-h_4+h_5, &&s_5(h_5)=h_4-h_5+h_6,&\\
 &s_6(h_6)=h_5-h_6+h_7, &&s_7(h_7)=h_0+h_6-h_7,&\\
 &s_8(h_8)=h_3-h_8. &&&
\end{aligned}
\end{equation}
It is useful to express the actions of the reflections on the roots, and we obtain
\begin{equation}\label{eqns:action_E8_alpha}
\begin{split}
 &s_0:(\al_0,\al_7)\mapsto(-\al_0,\al_7+\al_0),\qquad
 s_1:(\al_1,\al_2)\mapsto(-\al_1,\al_2+\al_1),\\
 &s_2:(\al_1,\al_2,\al_3)\mapsto(\al_1+\al_2,-\al_2,\al_3+\al_2),\\
 &s_3:(\al_2,\al_3,\al_4,\al_8)\mapsto(\al_2+\al_3,-\al_3,\al_4+\al_3,\al_8+\al_3),\\
 &s_4:(\al_3,\al_4,\al_5)\mapsto(\al_3+\al_4,-\al_4,\al_5+\al_4),\\
 &s_5:(\al_4,\al_5,\al_6)\mapsto(\al_4+\al_5,-\al_5,\al_6+\al_5),\\
 &s_6:(\al_5,\al_6,\al_7)\mapsto(\al_5+\al_6,-\al_6,\al_7+\al_6),\\
 &s_7:(\al_6,\al_7)\mapsto(\al_6+\al_7,-\al_7),\qquad
 s_8:(\al_3,\al_8)\mapsto(\al_3+\al_8,-\al_8).
\end{split}
\end{equation}
Note that fundamental weights and simple roots which are not explicitly shown in Equation \eqref{eqn:s_h_E8} and \eqref{eqns:action_E8_alpha} remain unchanged under the action of the corresponding reflections.

Under the linear actions on the weight lattice \eqref{eqn:s_weight_E8},
$W(E_8^{(1)})=\langle s_0,\dots,s_8\rangle$ forms an affine Weyl group of type $E_8^{(1)}$. 
Indeed, the following fundamental relations hold:
\begin{equation}\label{eqns:fundamental_We8}
 \hspace{-3em}(s_is_j)^{l_{ij}}=1,\ \text{where}\ 
 l_{ij}=
\begin{cases}
 1,& i=j\\
 3, &i=j-1,\ j=2,\dots,7,\  \text{or}\  (i,j)=(3,8),(7,0)\\
 2, &\text{otherwise}.
\end{cases} 
\end{equation}
Note that representing the simple reflections $s_i$ by nodes and connecting $i$-th and $j$-th nodes by $(l_{ij}-2)$ lines,
we obtain the Dynkin diagram of type $E_8^{(1)}$ shown in Figure \ref{fig:dynkin_E8}.

\begin{figure}[htbp]
\begin{center}
\includegraphics[width=0.6\textwidth]{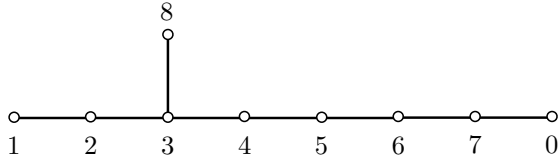}
\caption{Dynkin diagram of type $E_8^{(1)}$.}
\label{fig:dynkin_E8}
\end{center}
\end{figure}

\begin{remark}
We can also define the action of $W(E_8^{(1)})$ on the coroot lattice $\oc Q$ by replacing $\al_i$ with $\oc\al_i$ in the actions \eqref{eqns:action_E8_alpha}.  
Note that $\oc\de$ is invariant under the action of $W(E_8^{(1)})$.
Then, the transformations in $W(E_8^{(1)})$ preserve the form $\langle\cdot,\cdot\rangle$, that is, the following holds:
\begin{equation}\label{eqn:preserve_bracket}
 w.\langle\ga,\la\rangle
 =\langle w(\ga),w(\la)\rangle
 =\langle\ga,\la\rangle,
\end{equation}
for arbitrary $w\in W(E_8^{(1)})$, $\ga\in\oc Q$ and $\la\in P$.
\end{remark}

For each root $\al\in Q$, we define the Kac translation $T_\al$ on the weight lattice $P$ by
\begin{equation}
 T_\al(\la)=\la-\langle\oc\de,\la\rangle\al,\quad
 \la\in P,
\end{equation}
and on the coroot lattice $Q$ by
\begin{equation}\label{eqn:def_kac_coroot}
 T_\al(\oc\la)=\oc\la+\langle\oc\la,\al\rangle\oc\de,\quad
 \oc\la\in Q.
\end{equation}
We can easily verify that under the linear actions above the translation $T_\al$ have the following properties.
\begin{enumerate}
\item
For any $\al,\beta\in Q$, $T_\al\circ T_\beta=T_{\al+\beta}$.
\item
For any $w\in W(E_8^{(1)})$ and $\al\in Q$,
\begin{equation}\label{eqn:wT=Tw}
 w\circ T_\al=T_{w(\al)}\circ w.
\end{equation}
\item
For any $\al,\beta\in Q$, $T_\al(\beta)=\beta$.
\end{enumerate}
For any $\al=\sum_{i=0}^8c_i \al_i\in Q$, we define the squared length of $T_\al$ by
\begin{equation}\label{eqn:def_squared length_weight}
 |T_{\al}|^2:=\langle \oc\al,\al\rangle,
\end{equation}
where $\oc\al=\sum_{i=0}^8c_i \oc\al_i$.
Then, from \eqref{eqn:preserve_bracket} and \eqref{eqn:wT=Tw}, we obtain the following lemma.
\begin{lemma}
For any $\al,\beta\in Q$,
if $T_\al$ and $T_\beta$ are conjugate to each other in $W(E_8^{(1)})$, 
then the squared lengths of $T_\al$ and $T_\beta$ are equal. 
\end{lemma}
\begin{proof}
Assume that $T_\al$ and $T_\beta$, where $\al,\beta\in Q$, are conjugate to each other in $W(E_8^{(1)})$.
Let $w\circ T_\al\circ w^{-1}=T_\beta$,
where $w\in W(E_8^{(1)})$.
Since of \eqref{eqn:wT=Tw}, we obtain
$w\circ T_\al\circ w^{-1}=T_{w(\al)}$,
which gives
\begin{equation}
 T_{w(\al)}=T_\beta.
\end{equation}
Since of \eqref{eqn:preserve_bracket} and \eqref{eqn:def_squared length_weight}, the statement follows from
\begin{equation}
 |T_{\al}|^2=|T_{w(\al)}|^2=|T_{\beta}|^2.
\end{equation}
\end{proof}

In \cite{MSY2003:MR1958273} Murata {\em et al} consider the following translation as the time evolution of the Sakai's elliptic difference equation:
\begin{equation}\label{sakai_T}
T^{(M)}=s_{1238432543865432765438076543212345670834567234568345234832},
\end{equation}
whose action on the coroot lattice $\oc Q$ is given by
\begin{equation}\label{eqn:MSY_T_coroot}
 T^{(M)}(\oc\al_1)=\oc\al_1+2\oc\de,\quad
 T^{(M)}(\oc\al_2)=\oc\al_2-\oc\de,
\end{equation}
while in \cite{AHJN2016:MR3509963,JN2017:MR3673467} Joshi {\em et al} showed that 
the squared time evolution of the RCG equation corresponds to the following translation:
\begin{equation}\label{JN_T}
 T^{(JN)}=s_{56453483706756452348321 56453483706756452348321706734830468}^2,
\end{equation}
whose action on the coroot lattice $\oc Q$ is given by
\begin{equation}\label{eqn:RCG_T_coroot}
 T^{(JN)}(\oc\al_1)=\oc\al_1-2\oc\de,\quad
 T^{(JN)}(\oc\al_5)=\oc\al_5+\oc\de.
\end{equation}
Note that for convenience we use the following notations for the composition of the reflections $s_i$:
\begin{align}\label{eqn:notation_composition_si}
 s_{i_1\cdots i_m}=s_{i_1}\circ\dots\circ s_{i_m},\quad
 s_{i_1\cdots i_m}^2=s_{i_1\cdots i_m}\circ s_{i_1\cdots i_m},
\end{align}
where $i_1\cdots i_m\in\{0,\dots,8\}$.
Comparing \eqref{eqn:def_kac_coroot} and \eqref{eqn:MSY_T_coroot} 
and comparing \eqref{eqn:def_kac_coroot} and \eqref{eqn:RCG_T_coroot},
we can respectively express the translations $T^{(M)}$ and $T^{(JN)}$ by the Kac translations as the following:
\begin{equation}
 T^{(M)}=T_{\al_1},\quad
 T^{(JN)}=T_{\al_0+2\al_2+4\al_3+4\al_4+4\al_5+3\al_6+2\al_7+2\al_8},
\end{equation}
where
\begin{equation}
 |T_{\al_1}|^2=2,\quad
 |T_{\al_0+2\al_2+4\al_3+4\al_4+4\al_5+3\al_6+2\al_7+2\al_8}|^2=4.
\end{equation}
Therefore, the translations $T^{(M)}$ and $T^{(JN)}$ are not conjugate to each other in $W(E_8^{(1)})$ and respectively correspond to a NV and a NNV.

\begin{remark}
As another example, we will show the lattice and transformation group of type $A_4^{(1)}$ in \S \ref{appendix:A4}.
\end{remark}

\section{The initial value space of the RCG equation}
\label{section:General setting}
Consider the RCG equation \eqref{eqn:RCGeqn} as a discrete dynamical system. It is reversible and so the system maps $(x_n, y_n)$ to $(x_{n+1}, y_{n+1})$ at each forward time step or to $(x_{n-1}, y_{n-1})$ at each backward step. Because the coefficients of $y_{n+1}$ or $x_{n+1}$ may vanish, the iterates may become unbounded and we compactify this system by embedding it in the projective space $\mathbb P^1\times \mathbb P^1$.

Compactification is not enough to avoid all problems. Let $\phi$ and $\phi^{-1}$ be respectively the forward and backward time evolution of Equation \eqref{eqn:RCGeqn}.
We denote the action of these mappings $\phi$ by
\begin{subequations}
\begin{align}
 \phi\,:(x,y;\gae,\gao,z_0)&\mapsto \big(\tx,\ty;\gae,\gao,z_0+2(\gae+\gao)\big),\\
 \phi^{-1}\,:(x,y;\gae,\gao,z_0)&\mapsto \big(\utilde{x},\utilde{y};\gae,\gao,z_0-2(\gae+\gao)\big).
\end{align}
\end{subequations}
(We obtain Equation \eqref{eqn:RCGeqn} by writing
$x_n=\phi^n(x)$, $y_n=\phi^n(y)$, $z_n=\phi^n(z_0)$.)
Then there exist points where the image values of $\phi$ or $\phi^{-1}$ approach $0/0$.  We refer to such points as {\em base points} because they are equivalent to the usual definition used in the case of algebraic curves \cite{griffiths1989introduction}. Below, we describe the base points of the RCG equation explicitly.

In general, the process of blowing up a finite sequence of points $p_i$, possibly infinitely near each other, in $\mathbb P^1\times \mathbb P^1$, leads to a variety $X$, called the initial value space. Assuming there are 8 blow-ups, then let the sequence of blow-ups be $\pi_i : X_i\mapsto X_{i-1}$ of $p_i$ in $X_{i-1}$, with $X=X_8\to \ldots \to X_0=$ $\mathbb P^1\times \mathbb P^1$. Each blow-up replaces $p_i$ by an exceptional line $\mathcal{E}_i$. We refer to the total sequence of blow ups by $\ep: X \to \bbP^1\times\bbP^1$, and moreover, denote the linear equivalence classes of the total transform of vertical and horizontal lines in $\bbP^1\times\bbP^1$  respectively by $H_0$ and $H_1$.

To find base points of Equation \eqref{eqn:RCGeqn}, we need to find simultaneous zeroes of pairs of polynomials. For example, base points arising from the component $\tilde y$ in the action of $\phi$ lie at the simultaneous solutions of 
\[
\begin{cases}
(1-k^2\sz^4)\cge\dge\, xy-(\cge^2-\cz^2)\cz\,\dz-(1-k^2\sge^2 \sz^2)\cz\,\dz\, x^2=0,\\
k^2(\cge^2-\cz^2)\cz\,\dz\, x^2 y-(1-k^2 \sz^4)\cge\dge\, x+(1-k^2\sge^2\sz^2)\cz\,\dz\,y=0,
\end{cases}
\]
where the dependence on $n$ has been suppressed. These  polynomial equations can be solved explicitly. For example, the first equation in the above pair can be solved for $y$ in terms of $x$, and substituting into the second equation leads to a quartic equation for $x$. The four solutions of this equation can be expressed explicitly in terms of the coefficient elliptic functions by using elliptic function identities. A similar argument gives us four more base points arising from the remaining equations; for $\tx$ in the mapping $\phi$ and for $\utilde{x}$ and $\utilde{y}$ in the mapping $\phi^{-1}$.

These lead us to the eight base points listed below:
\begin{subequations}\label{eqn:basepoints_RCG}
\begin{align}
 &p_1:(x,y)=\big(\cd{\gao+\ka},\cd{z_0-\gae-\gao+\ka}\big),\\
 &p_2:(x,y)=\big(\cd{\gao+\iii K'},\cd{z_0-\gae-\gao+\iii K'}\big),\\
 &p_3:(x,y)=\big(\cd{\gao+2K},\cd{z_0-\gae-\gao+2K}\big),\\
 &p_4:(x,y)=\big(\cd{\gao},\cd{z_0-\gae-\gao}\big),\\
 &p_5:(x,y)=\big(\cd{z_0+\ka},\cd{\gae+\ka}\big),\\
 &p_6:(x,y)=\big(\cd{z_0+\iii K'},\cd{\gae+\iii K'}\big),\\
 &p_7:(x,y)=\big(\cd{z_0+2K},\cd{\gae+2K}\big),\\
 &p_8:(x,y)=\big(\cd{z_0},\cd{\gae}\big),
\end{align}
\end{subequations}
where $K=K(k)$ and $K'=K'(k)$ are complete elliptic integrals and
\begin{equation}\label{eqn:kappa}
 \ka=2K+\iii K',
\end{equation}
which lie on the elliptic curve
\begin{equation}\label{eqn:RCG_curve}
 \sn{z_0-\gae}^2(1+k^2x^2y^2)+2\cn{z_0-\gae}\dn{z_0-\gae}xy-(x^2+y^2)=0.
\end{equation}

The base points \eqref{eqn:basepoints_RCG} can be generalized to 
\begin{equation}\label{eqn:basepoints}
 p_i:(x,y)=\big(\cd{c_i+\eta},\cd{\eta-c_i}\big),\quad i=1,\dots,8,
\end{equation}
where $c_i$, $i=1,\dots,8$, and $\eta$ are non-zero complex parameters. 
These points lie on the elliptic curve
\begin{equation}\label{eqn:Jelliptic_curve}
 \sn{2\eta}^2(1+k^2x^2y^2)+2\cn{2\eta}\dn{2\eta}xy-(x^2+y^2)=0.
\end{equation}
The generalized base points \eqref{eqn:basepoints} and elliptic curve \eqref{eqn:Jelliptic_curve}
can be respectively reduced to the points \eqref{eqn:basepoints_RCG} and curve \eqref{eqn:RCG_curve} by taking
\begin{subequations}\label{eqns:condition_RCG}
\begin{align} 
 &\cc2=\cc1+2K,\quad
 \cc3=\cc1+\iii K',\quad
 \cc4=\cc1+\ka,\quad
 \cc6=\cc5+2K,\\
 &\cc7=\cc5+\iii K',\quad
 \cc8=\cc5+\ka,
\end{align}
and letting
\begin{equation}
 z_0=\eta+\cc5+\ka,\quad
 \gae=\cc5-\eta+\ka,\quad
 \gao=\eta+\cc1+\ka.
\end{equation}
\end{subequations}
Note that the two biquadratic curves \eqref{eqn:RCG_curve} and \eqref{eqn:Jelliptic_curve} are non-singular, for $k\not=0, \pm 1$. It follows that each curve is parametrized by elliptic functions. The resulting initial value space obtained after resolution is of type $A_0^{(1)}$ as shown in the following section. 

\section{Cremona isometries}
\label{section:cremona}
As explained in \S\ref{section:General setting}, we now investigate a variety $X$, obtained after a sequence of blow-ups. We focus on surfaces $X$ defined by blowing up  base points on biquadratic curves, such as Equations \eqref{eqn:basepoints} and \eqref{eqn:Jelliptic_curve}. The resulting structure contains equivalence classes of lines and information about their  intersections. In this section, we construct Cremona isometries, which are roughly speaking, mappings of $X$ that preserve this structure. As a result, they provide symmetries of the dynamical system iterated on $X$.

An important object in this framework is given by the Picard lattice of $X$, or ${\rm Pic}(X)$, which is defined by
\begin{equation}
 \PicX=\bbZ H_0+\bbZ H_1+\bbZ \EE{1}+\cdots+\bbZ \EE{8},
\end{equation}
where $\EE{i}=\ep^{-1}(p_i)$, $i=1,\dots,8$, are exceptional divisors obtained from blow-up of the base points \eqref{eqn:basepoints}. We define a symmetric bilinear form, called the intersection form, on ${\rm Pic}(X)$ by 
\begin{equation}
 (H_i|H_j)=1-\de_{ij},\quad
 (H_i|\EE{j})=0,\quad
 (\EE{i}|\EE{j})=-\de_{ij},
\end{equation}
where $\de_{ij}=0$, if $i\not=j$, or $1$, if $i=j$. The anti-canonical divisor of $X$ is given by 
\begin{equation}
 -K_X=2H_0+2H_1-\sum_{i=1}^8\EE{i}.
\end{equation}
For later convenience, let $\delta=-K_X$.
The anti-canonical divisor $\de$ corresponds to the curve of bi-degree $(2,2)$ passing through the base points $p_i$, $i=1,\dots,8$, with multiplicity $1$, that is, the curve \eqref{eqn:Jelliptic_curve}.
Since this curve is non-singular, for $k\not=0, \pm 1$,
the anti-canonical divisor cannot be decomposed.
Therefore, we can identify the surface $X$ as being of type $A_0^{(1)}$ in Sakai's classification\cite{SakaiH2001:MR1882403}.

We define the root lattice
\begin{equation}
 Q(A_0^{(1)\bot})=\sum_{k=0}^8\bbZ\beta_k
\end{equation}
in {\rm Pic}(X) 
that are orthogonal to the anti-canonical divisor $\delta$. 
The simple roots $\beta_i$, $i=0,\dots,8$, are given by
\begin{align}
 &\beta_1=H_1-H_0,\quad
 \beta_2=H_0-\EE1-\EE2,\quad
 \beta_i=\EE{i-1}-\EE{i},\quad i=3,\dots,7,\notag\\
 &\beta_8=\EE1-\EE2,\quad
 \beta_0=\EE7-\EE8,
\end{align}
where
\begin{equation}
 \de=2\beta_1+4\beta_2+6\beta_3+5\beta_4+4\beta_5+3\beta_6+2\beta_7+3\beta_8+\beta_0.
\end{equation}
We can easily verify that
\begin{equation}
 (\beta_i|\beta_j)
 =\begin{cases}
 -2,& i=j\\
 \ 1, &i=j-1\quad (j=2,\dots,7),\quad \text{or\hspace{0.5em}if}\quad (i,j)=(3,8),(7,0)\\
 \ 0, &\text{otherwise}.
\end{cases}
\end{equation}
Representing intersecting $\beta_i$ and $\beta_j$ by a line between nodes $i$ and $j$, we obtain the Dynkin diagram of $E_8^{(1)}$ shown in Figure \ref{fig:dynkin_E8}.
\begin{remark}
From {\rm Pic}$(X)$ we can obtain the coroot lattice, weight lattice and root lattice in \S \ref{section:weight_lattice_E8}.
Indeed, the coroot lattice $\oc Q$ is given from {\rm Pic}$(X)$ by
\begin{equation}
 \oc\al_i=\beta_i,\quad i=0,\dots,8,
\end{equation}
the weight lattice $P$ is given from {\rm Pic}$(X)/(\bbZ\de)$ by
\begin{subequations}
\begin{align}
 &h_0\equiv -\mathcal{E}_8,\quad
 h_1\equiv -H_0,\quad
 h_2\equiv -H_0-H_1,\\
 &h_3\equiv -\mathcal{E}_3-\mathcal{E}_4-\mathcal{E}_5-\mathcal{E}_6-\mathcal{E}_7-\mathcal{E}_8,\\
 &h_4\equiv -\mathcal{E}_4-\mathcal{E}_5-\mathcal{E}_6-\mathcal{E}_7-\mathcal{E}_8,\quad
 h_5\equiv -\mathcal{E}_5-\mathcal{E}_6-\mathcal{E}_7-\mathcal{E}_8,\\
 &h_6\equiv -\mathcal{E}_6-\mathcal{E}_7-\mathcal{E}_8,\quad
 h_7\equiv -\mathcal{E}_7-\mathcal{E}_8,\\
 &h_8\equiv -H_0-H_1+\mathcal{E}_1,
\end{align}
\end{subequations}
and the root lattice $Q$ is given from {\rm Pic}$(X)/(\bbZ\de)$ by
\begin{equation}
 \al_i\equiv\beta_i,\quad i=0,\dots,8.
\end{equation}
Note that in this case we define the bilinear pairing $\langle\cdot,\cdot\rangle:\oc Q\times P\to\bbZ$\, 
for arbitrary $\oc\al\in\oc Q$ and $h\in P$ by
\begin{equation}
 \langle\oc\al,h\rangle=-(\oc\al|h).
\end{equation}
\end{remark}

Therefore, in a similar manner as in \S \ref{section:weight_lattice_E8}, we can define the transformation group $W(E_8^{(1)})$ as follows.

\begin{definition}[\cite{DO1988:MR1007155}]
\label{def:cremona}
An automorphism of {\rm Pic}(X) is called a Cremona isometry if it preserves 
\begin{enumerate}
\item
the intersection form $(\,|\,)$ on {\rm Pic}(X);
\item
the canonical divisor $K_X$;
\item
effectiveness of each effective divisor of {\rm Pic}(X).
\end{enumerate}
\end{definition}

It is well-known that the reflections are Cremona isometries.
In this case we define the reflections $s_i$, $i=0,\dots,8$, by the following linear actions:
\begin{equation}\label{eqn:def_si}
 s_i(v)=v+(v|\,\beta_i)\,\beta_i,
\end{equation}
for all $v\in \PicX$.
They collectively form an affine Weyl group of type $\EE8^{(1)}$, denoted by $\WEe$. 
Namely, we can easily verify that under the actions \eqref{eqn:def_si} the fundamental relations \eqref{eqns:fundamental_We8} hold.
Moreover, for each root $\beta\in Q(A_0^{(1)\bot})$, we can define the Kac translation $T_\beta$ on the Picard lattice by
\begin{equation}
 T_\beta(\la)=\la+(\de|\la)\beta-\left(\frac{(\beta|\beta)(\de|\la)}{2}+(\beta|\la)\right)\de,\quad
 \la\in \PicX.
\end{equation}

\section{Birational actions of the Cremona isometries for the Jacobi's setting}
\label{section:birational}
In this section, we give the birational actions of the Cremona isometries on the coordinates and parameters of the base points \eqref{eqn:basepoints}.
By using these birational actions, we reconstruct Equation \eqref{eqn:RCGeqn}. 

We focus on a particular example first to explain how to deduce such birational actions. 
Recall $H_0$ and $H_1$ are given by the linear equivalence classes of vertical lines $x=\text{constant}$ and horizontal lines $y=\text{constant}$, respectively. 
Applying the reflection operator $s_2$ given by \eqref{eqn:def_si} to $H_1$, we find $s_2(H_1)=H_0+H_1-\mathcal{E}_1-\mathcal{E}_2$,
which means that $s_2(y)$ can be described by the curve of bi-degree $(1,1)$ passing through base points $p_1$ and $p_2$ with multiplicity $1$.
 (See \cite{KNY2017:MR3609039} for for more detail.)
This result leads us to the birational action given below in Equation \eqref{eqn:action_E8_J_y}.
Similarly, from the linear actions of $s_i$, $i=0,\dots,8$, we obtain their birational actions on the coordinates and parameters of the base points \eqref{eqn:basepoints} as follows.
The actions of the generators of $\WEe$ on the coordinates $(x,y)$ are given by
\begin{subequations}\label{eqns:action_E8_J_para_xy}
\begin{align}
 &s_1(x)=y,\quad
 s_1(y)=x,\\
 &\bfrac{s_2(y)-\cd{2\eta-\frac{\cc1-\cc2}{2}}}{s_2(y)-\cd{2\eta+\frac{\cc1-\cc2}{2}}}
 \bfrac{x-\cd{\eta+\cc1}}{x-\cd{\eta+\cc2}}
 \bfrac{y-\cd{\eta-\cc2}}{y-\cd{\eta-\cc1}}\notag\\
 &=\bfrac{1-\dfrac{\cd{\eta-\cc2}}{\cd{\eta}}}{1-\dfrac{\cd{\eta-\cc1}}{\cd{\eta}}}
 \bfrac{1-\dfrac{\cd{\eta+\cc1}}{\cd{\eta}}}{1-\dfrac{\cd{\eta+\cc2}}{\cd{\eta}}}
 \bfrac{1-\dfrac{\cd{2\eta-\frac{\cc1-\cc2}{2}}}{\cd{\frac{\cc1+\cc2}{2}}}}{1-\dfrac{\cd{2\eta+\frac{\cc1-\cc2}{2}}}{\cd{\frac{\cc1+\cc2}{2}}}},
 \label{eqn:action_E8_J_y}
\end{align}
while those on the parameters $\cc{i}$, $i=1,\dots,8$, and $\eta$ are given by
\begin{align}
 &s_0(\cc7)=\cc8,\quad
 s_0(\cc8)=\cc7,\quad
 s_1(\eta)=-\eta,\quad
 s_2(\eta)=\eta-\frac{2\eta+\cc1+\cc2}{4},\\
 &s_2(\cc{i})=
 \begin{cases}
 \cc{i}-\dfrac{3(2\eta+\cc1+\cc2)}{4},& i=1,2,\\
 \cc{i}+\dfrac{2\eta+\cc1+\cc2}{4},& i=3,\dots,8,
 \end{cases}\\
 &s_k(\cc{k-1})=\cc{k},\quad 
 s_k(\cc{k})=\cc{k-1},\quad
 k=3,\dots,7,\\
 &s_8(\cc1)=\cc2,\quad
 s_8(\cc2)=\cc1.
\end{align}
\end{subequations}
Note that $\la=\sum_{i=1}^8\cc{i}$ is invariant under the action of $\WEe$.

For Jacobi's elliptic function $\cd{u}$ it is well known that shifts by half periods give the following relations:
\begin{equation}
 \cd{u+2K}=-\cd{u},\quad
 \cd{u+\iii K'}=\dfrac{1}{k\, \cd{u}}.
\end{equation}
These identities motivate our search for the transformations that are identity mappings on the $\PicX$.
Indeed, we define such transformations $\io_i$, $i=1,\dots,4$, by the following actions:
\begin{subequations}\label{eqns:action_iota}
\begin{align}
 &\io_1:(\cc{1},\dots,\cc{8},\eta,x,y)
 \mapsto\left(\cc{1}-\frac{\iii K'}{2},\dots,\cc{8}-\frac{\iii K'}{2},\eta-\frac{\iii K'}{2},\frac{1}{kx},y\right),\\
 &\io_2:(\cc{1},\dots,\cc{8},\eta,x,y)
 \mapsto\left(\cc{1}-\frac{\iii K'}{2},\dots,\cc{8}-\frac{\iii K'}{2},\eta+\frac{\iii K'}{2},x,\frac{1}{ky}\right),\\
 &\io_3:(\cc{1},\dots,\cc{8},\eta,x,y)
 \mapsto\left(\cc{1}-K,\dots,\cc{8}-K,\eta-K,-x,y\right),\\
 &\io_4:(\cc{1},\dots,\cc{8},\eta,x,y)
 \mapsto\left(\cc{1}-K,\dots,\cc{8}-K,\eta+K,x,-y\right).
\end{align}
\end{subequations}
Adding the transformations $\io_i$, we extend $\WEe$ to
\begin{equation}
 \tWEe=\lrangle{\io_1,\io_2,\io_3,\io_4}\rtimes\WEe.
\end{equation}
In general, for a function $F=F(c_i,\eta,x,y)$, we let an element
$w\in\tWEe$ act as $w.F=F(w.c_i,w.\eta,w.x,w.y)$, that is, 
$w$ acts on the arguments from the left. 
Under the birational actions \eqref{eqns:action_E8_J_para_xy} and \eqref{eqns:action_iota}, 
the generators of $\tWEe$ satisfy 
the fundamental relations of type $E_8^{(1)}$ \eqref{eqns:fundamental_We8}  and the following relations:
\begin{subequations}\label{eqns:relation_iota}
\begin{align}
 &(\io_i\io_j)^{m_{ij}}=1,\quad i,j=1,2,3,4,\quad
 \io_ks_l=s_l\io_k,\quad k=1,2,3,4,~ l\neq 1,2,\\
 &\io_{\{1,2,3,4\}}s_1=s_1\io_{\{2,1,4,3\}},\quad
 \io_1s_2=s_2\io_1\io_2,\quad
 \io_2s_2=s_2\io_2,\\
 &\io_3s_2=s_2\io_3\io_4,\quad
 \io_4s_2=s_2\io_4,
\end{align}
\end{subequations}
where 
\begin{equation}
m_{ij}=
\begin{cases}
 1,& i=j\\
 2, &\text{otherwise}.
\end{cases}
\end{equation}

Now we are in a position to derive Equation \eqref{eqn:RCGeqn} from the Cremona transformations.
Note that for convenience we use the notation \eqref{eqn:notation_composition_si} for the composition of the reflections $s_i$ and the notation
\begin{equation}\label{notation:c}
 \cc{j_1\cdots j_n}=\cc{j_1}+\cdots+\cc{j_n},\quad j_1\cdots j_n\in\{1,\dots,8\},
\end{equation}
for the summation of the parameters $\cc{i}$.
Let
\begin{equation}
 \RJ1=s_{56453483706756452348321 56453483706756452348321706734830468}\io_4\io_3\io_2\io_1.
\end{equation}
The action of $\RJ1$ on the root lattice $Q(A_0^{(1)\bot})$:
\begin{equation}
 \RJ1:
 \begin{pmatrix}
 \al_0\\\al_1\\\al_2\\\al_3\\\al_4\\\al_5\\\al_6\\\al_7\\\al_8
 \end{pmatrix}
 \mapsto
\left(\begin{array}{ccccccccc}
 -1 & 0 & 0 & 0 & 0 & 0 & 0 & 0 & 0 \\
 -1 & -1 & -4 & -6 & -5 & -4 & -3 & -2 & -3 \\
 0 & 0 & 1 & 2 & 1 & 0 & 0 & 0 & 1 \\
 0 & 0 & 0 & -1 & 0 & 0 & 0 & 0 & 0 \\
 0 & 0 & 0 & 0 & -1 & 0 & 0 & 0 & 0 \\
 1 & 1 & 2 & 4 & 4 & 3 & 3 & 2 & 2 \\
 0 & 0 & 0 & 0 & 0 & 0 & -1 & 0 & 0 \\
 0 & 0 & 0 & 0 & 0 & 0 & 0 & -1 & 0 \\
 0 & 0 & 0 & 0 & 0 & 0 & 0 & 0 & -1 \\
\end{array}\right)
 \begin{pmatrix}
 \al_0\\\al_1\\\al_2\\\al_3\\\al_4\\\al_5\\\al_6\\\al_7\\\al_8
 \end{pmatrix},
\end{equation}
is not translational, and that on the parameter space:
\begin{align}
 &\RJ1(\cc{i})=-\cc{i}+\frac{\cc{1234}-\cc{5678}}{4}-\ka,\quad i=1,\dots,4,\\
 &\RJ1(\cc{j})=-\cc{j}+\frac{\cc{1234}+3\cc{5678}}{4}-\ka,\quad j=5,\dots,8,\quad
 \RJ1(\eta)=\eta+\frac{\la}{2},
\end{align}
where $\ka$ is defined by \eqref{eqn:kappa}, is also not translational.
However, when the parameters take special values \eqref{eqns:condition_RCG},
the action of $\RJ1$ becomes the translational motion in the parameter subspace:
\begin{equation}\label{eqn:action_RJ1_special_para}
 \RJ1:(\gae,\gao,z_0)\mapsto(\gae,\gao,z_0+2(\gae+\gao)-2\ka),
\end{equation}
and then the action on the coordinates $(x,y)$ with 
$x_n=\RJ1^n(x)$, $y_n=\RJ1^n(y)$, $z_n=\RJ1^n(z_0)$,
gives the time evolution of Equation \eqref{eqn:RCGeqn}, that is, $\RJ1=\phi$.
Note that we can without loss of generality ignore ``$2\ka$" in 
\begin{equation}
 \RJ1(z_0)=2(\gae+\gao)-2\ka,
\end{equation}
since of the form of Equation \eqref{eqn:RCGeqn} and the following relations: 
\begin{equation}
 \sn{u+2\ka}=\sn{u},\quad
 \cn{u+2\ka}=-\cn{u},\quad
 \dn{u+2\ka}=-\dn{u}.
\end{equation}

\section{Special solutions of the RCG equation}
\label{section:Special sols RCG equation}
In this section, we show the special solutions of the RCG equation.

Let us consider Equation \eqref{eqn:RCGeqn} under the following condition:
\begin{equation}
 \gao=\dfrac{\iii K'}{2},
\end{equation}
which gives
\begin{equation}
 \sgo=\dfrac{\iii}{k^{1/2}},\quad
 \cgo=\dfrac{(1+k)^{1/2}}{k^{1/2}},\quad
 \dgo=(1+k)^{1/2}.
\end{equation}
Then, the base points $p_i$, $i=1,2,3,4$, in \eqref{eqn:basepoints_RCG} can be expressed by
\begin{subequations}
\begin{align}
 &p_1:(x,y)=\left(-\dfrac{1}{k^{1/2}},-\cd{z_0-\gae+\dfrac{\iii K'}{2}}\right),\\
 &p_2:(x,y)=\left(\dfrac{1}{k^{1/2}},\cd{z_0-\gae+\dfrac{\iii K'}{2}}\right),\\
 &p_3:(x,y)=\left(-\dfrac{1}{k^{1/2}},-\cd{z_0-\gae-\dfrac{\iii K'}{2}}\right),\\
 &p_4:(x,y)=\left(\dfrac{1}{k^{1/2}},\cd{z_0-\gae-\dfrac{\iii K'}{2}}\right).
\end{align}
\end{subequations}
This means that 
there exist the bi-degree $(1,0)$ curve $x=-k^{-1/2}$, passing through $p_1$ and $p_3$, and  the bi-degree $(1,0)$ curve $x=k^{-1/2}$, passing through $p_2$ and $p_4$,
which correspond to $H_0-\mathcal{E}_1-\mathcal{E}_3$ and $H_0-\mathcal{E}_2-\mathcal{E}_4$, respectively.
Moreover, the action
\begin{equation}
 \phi:~H_0-\mathcal{E}_1-\mathcal{E}_3\leftrightarrow H_0-\mathcal{E}_2-\mathcal{E}_4,
\end{equation}
implies that there exist the special solutions when
\begin{equation}
 x_n=\pm\dfrac{(-1)^n}{k^{1/2}}.
\end{equation}
Therefore, we obtain the following lemma.
\begin{lemma}\label{lemma:special_sol}
The following are special solutions of Equation \eqref{eqn:RCGeqn}$:$
\begin{subequations}\label{eqns:speacil_sol_1}
\begin{align}
 &(x_n,y_n)=\left(\dfrac{(-1)^n}{k^{1/2}},\dfrac{(-1)^n}{k^{1/2}}\right),
 &&(x_n,y_n)=\left(\dfrac{(-1)^n}{k^{1/2}},\dfrac{(-1)^{n+1}}{k^{1/2}}\right),\\
 &(x_n,y_n)=\left(\dfrac{(-1)^{n+1}}{k^{1/2}},\dfrac{(-1)^n}{k^{1/2}}\right),
 &&(x_n,y_n)=\left(\dfrac{(-1)^{n+1}}{k^{1/2}},\dfrac{(-1)^{n+1}}{k^{1/2}}\right),
\end{align}
\end{subequations}
and
\begin{equation}
 (x_n,y_n)=\left(\dfrac{(-1)^n}{k^{1/2}},\dfrac{\iii\,\tan(u_n)}{k^{1/2}}\right),\quad
 (x_n,y_n)=\left(\dfrac{(-1)^{n+1}}{k^{1/2}},-\dfrac{\iii\,\tan(u_n)}{k^{1/2}}\right),
\end{equation}
where $\iii=\sqrt{-1}$ and $u_n$ is the solution of the following linear equation:
\begin{equation}\label{eqn:linear_un}
 u_{n+1}+u_n=\tan^{-1}\left(-\iii\,\dfrac{(1-k\,\sge^2)\cz_n\dz_n}{\cge \dge(1-k\, \sz_n^2)}\right).
\end{equation}
\end{lemma}
\begin{proof}
Under the condition 
\begin{equation}\label{eqn:xn_(-1)n}
 x_n=\dfrac{(-1)^n}{k^{1/2}},
\end{equation}
Equation \eqref{eqn:RCGeqn} are reduced to the following discrete Riccati equation:
\begin{equation}\label{eqn:riccati_y}
 y_{n+1}+y_n=\dfrac{\iii\,A_n}{k^{1/2}}(1+k\, y_{n+1}y_n),
\end{equation}
where $A_n$ is given by 
\begin{equation}\label{eqn:def_An}
 A_n=-\iii\,\dfrac{(1-k\,\sge^2)\cz_n\dz_n}{\cge \dge(1-k\, \sz_n^2)}.
\end{equation}

Note that under the condition 
\begin{equation}
 x_n=\dfrac{(-1)^{n+1}}{k^{1/2}},
\end{equation}
Equation \eqref{eqn:RCGeqn} can be reduced to 
\begin{equation}
 y_{n+1}+y_n=-\dfrac{\iii\,A_n}{k^{1/2}}(1+k\, y_{n+1}y_n),
\end{equation}
which can be rewritten as the discrete Riccati equation \eqref{eqn:riccati_y} by the transformation $y_n\mapsto -y_n$.
Therefore, it is enough for us to just consider the case \eqref{eqn:xn_(-1)n}.

Let us consider the solutions of the discrete Riccati equation \eqref{eqn:riccati_y}.
If 
\begin{equation}
 1+k\, y_{n+1}y_n=0,
\end{equation}
then we obtain
\begin{equation}
 y_{n+1}+y_n=0.
\end{equation}
Therefore, we obtain 
\begin{equation}
 y_n=\dfrac{(-1)^n}{k^{1/2}},~\dfrac{(-1)^{n+1}}{k^{1/2}},
\end{equation}
which gives the special solutions \eqref{eqns:speacil_sol_1}.
In the following we assume
\begin{equation}
 1+k\, y_{n+1}y_n\neq0.
\end{equation}
Then, the discrete Riccati equation \eqref{eqn:riccati_y} can be rewritten as the following:
\begin{equation}\label{eqn:riccati_y_2}
 \dfrac{y_{n+1}+y_n}{1+k\, y_{n+1}y_n}=\dfrac{\iii\, A_n}{k^{1/2}}.
\end{equation}
By letting
\begin{equation}
 y_n=\dfrac{\iii\,\tan(u_n)}{k^{1/2}},
\end{equation}
and using the tangent addition formula, 
the discrete Riccati equation \eqref{eqn:riccati_y_2} can be rewritten as
\begin{equation}
 \tan(u_{n+1}+u_n)=A_n,
\end{equation}
which gives the linear equation \eqref{eqn:linear_un}.
Therefore, we have completed the proof.
\end{proof}

\section*{Acknowledgments}
This research was supported by an Australian Laureate Fellowship \# FL120100094 and grant \# DP160101728 from the Australian Research Council and JSPS KAKENHI Grant Number JP17J00092.
\appendix
\section{$A_4^{(1)}$-lattice}
\label{appendix:A4}
In this section, we give a more detailed description of the weight lattice and affine Weyl group by using the lattice of type $A_4^{(1)}$ as an example.

We consider the following $\bbZ$-modules:
\begin{equation}
 \oc Q=\sum_{k=0}^4\bbZ\oc\al_k,\quad
 P=\sum_{k=0}^4\bbZ h_k,
\end{equation}
with the bilinear pairing $\langle\cdot,\cdot\rangle:\oc Q\times P\to\bbZ$ defined by
\begin{equation}\label{eqn:def_dual_A4}
 \langle\oc\al_i,h_j\rangle=\de_{ij},\quad 0\leq i, j \leq 4.
\end{equation}
We also define the submodule of $P$ by
$Q=\sum_{k=0}^4\bbZ\al_k$,
where $\al_i$, $i=0,\dots,4$, are defined by
\begin{equation}\label{eqn:def_alpha_A4}
 \begin{pmatrix}
 \al_0\\ \al_1\\ \al_2\\ \al_3\\ \al_4
 \end{pmatrix} 
 =(A_{ij})_{i, j=0}^4
 \begin{pmatrix}
 h_0\\h_1\\h_2\\h_3\\h_4
 \end{pmatrix},
\end{equation}
and satisfy
\begin{equation}\label{eqn:root_coroot_A4}
 \langle\oc\al_i,\al_j\rangle=A_{ij}.
\end{equation}
Here, $(A_{ij})_{i, j=0}^4$ is the Generalized Cartan matrix of type $A_4^{(1)}$:
\begin{equation}\label{eqn:cartan_A4}
 (A_{ij})_{i, j=0}^4
 =\begin{pmatrix}
 2&-1&0&0&-1\\-1&2&-1&0&0\\0&-1&2&-1&0\\0&0&-1&2&-1\\-1&0&0&-1&2
 \end{pmatrix}.
\end{equation}
Then, the generators $\{\oc\al_0,\dots,\oc\al_4\}$, $\{h_0,\dots,h_4\}$ and $\{\al_0,\dots,\al_4\}$
are identified with simple coroots, fundamental weights and simple roots of type $A_4^{(1)}$, respectively.
We note that the following relation holds:
\begin{equation}\label{eqn:cond_root_al_A4}
 \al_0+\al_1+\al_2+\al_3+\al_4=0,
\end{equation}
and we call the corresponding coroot as null-coroot denoted by $\oc\de$:
\begin{equation}
 \oc\de=\oc\al_0+\oc\al_1+\oc\al_2+\oc\al_3+\oc\al_4.
\end{equation} 
In the following subsections, we consider the transformation group acting on theses lattices.
\subsection{Affine Weyl group of type $A_4^{(1)}$}
In this section, we consider the transformations which collectively form an affine Weyl group of type $A_4^{(1)}$.

We define the transformations $s_i$, $i=0,\dots,4$, by the reflections for the roots $\{\al_0,\dots,\al_4\}$:
\begin{equation}\label{eqn:s_weight_A4}
 s_i(\la)=\la-\langle\oc\al_i,\la\rangle\al_i,\quad i=0,\dots,4,\quad \la\in P,
\end{equation}
which give
\begin{equation}\label{eqn:s_h_A4}
\begin{split}
 &s_0(h_0)=-h_0+h_1+h_4,\quad
 s_1(h_1)=h_0-h_1+h_2,\quad
 s_2(h_2)=h_1-h_2+h_3,\\
 &s_3(h_3)=h_2-h_3+h_4,\quad
 s_4(h_4)=h_0+h_3-h_4.
\end{split}
\end{equation}
From definitions \eqref{eqn:def_dual_A4}, \eqref{eqn:def_alpha_A4}, \eqref{eqn:root_coroot_A4} and \eqref{eqn:s_weight_A4}, 
we can compute actions on the simple roots $\al_i$, $i=0,\dots,4$, as the following:
\begin{equation}\label{eqn:action_A4_alpha}
s_i(\al_j)
=\begin{cases}
 -\al_j,& i=j\\
 \al_j+\al_i, &i=j\pm1\quad \text{(mod 5)}\\
 \al_j, &i=j\pm2\quad \text{(mod 5)}.
\end{cases}
\end{equation}
Under the linear actions on the weight lattice \eqref{eqn:s_h_A4},
$W(A_4^{(1)})=\langle s_0,s_1,s_2,s_3,s_4\rangle$ forms an affine Weyl group of type $A_4^{(1)}$, that is,
the following fundamental relations hold:
\begin{equation}
 (s_is_j)^{l_{ij}}=1,\quad
 \text{where}\quad
 l_{ij}=
\begin{cases}
 1,& i=j\\
 3, &i=j\pm1\quad \text{(mod 5)}\\
 2, &i=j\pm2\quad \text{(mod 5)}.
\end{cases} 
\end{equation}
Note that representing the simple reflections $s_i$ by nodes and connecting $i$-th and $j$-th nodes by $(l_{ij}-2)$ lines,
we obtain the Dynkin diagram of type $A_4^{(1)}$ shown in Figure \ref{fig:dynkin_A4}.

\begin{figure}[htbp]
\begin{center}
\includegraphics[width=0.32\textwidth]{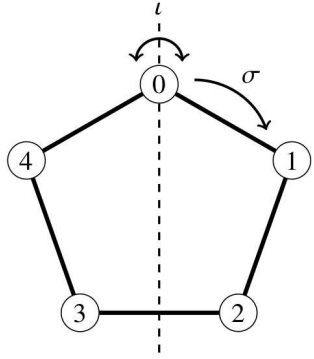}
\end{center}
\caption{Dynkin diagram of type $A_4^{(1)}$. The transformation $\iota$ is the reflection with respect to the dotted line, and the transformation $\si$ is the rotation symmetry with respect to an angle of $2\pi/5$ in a clockwise manner.}
\label{fig:dynkin_A4}
\end{figure}

\begin{remark}
We can also define the action of $W(A_4^{(1)})$ on the coroot lattice $\oc Q$ by replacing $\al_i$ with $\oc\al_i$ in the actions \eqref{eqn:action_A4_alpha}.  
Note that $\oc\de$ is invariant under the action of $W(A_4^{(1)})$.
Then, the transformations in $W(A_4^{(1)})$ preserve the form $\langle\cdot,\cdot\rangle$.
\end{remark}

Let $M_i$, $i=0,\dots,4$, be the orbits of $h_i$, $i=0,\dots,4$, defined by  
\begin{equation}\label{eqn:def_M_i}
 M_i=\left\{w(h_i)\,\left|\, w\in W(A_4^{(1)})\right\}\right.,
\end{equation}
and $T_i$, $i=0,\dots,4$, be the transformations defined by
\begin{subequations}\label{eqn:translation_Ti_A4}
\begin{align}
 &T_0=s_{04321234},\quad
 T_1=s_{10432340},\quad
 T_2=s_{21043401},\\
 &T_3=s_{32104012},\quad
 T_4=s_{43210123},
\end{align}
\end{subequations}
whose actions on the fundamental weights and simple roots are given by
\begin{equation}
 T_i(h_j)=h_j-\al_i,\quad
 T_i(\al_j)=\al_j,\quad
 \quad i,j=0,\dots,4.
\end{equation}
Note that $T_0\circ T_1\circ T_2\circ T_3\circ T_4=1$.
Then, the following lemma holds.

\begin{lemma}\label{lemma:orbits}
The following hold:
\begin{equation}
 M_i=\left\{h_i+\al\,\left|\, \al\in Q \right\}\right.,\quad i=0,\dots,4.
\end{equation}
\end{lemma}
\begin{proof}
The relation
$M_0\subset\left\{h_0+\al\,\left|\, \al\in Q \right\}\right.$
is obvious from the definition \eqref{eqn:s_weight_A4},
and the relation
$M_0\supset\left\{h_0+\al\,\left|\, \al\in Q \right\}\right.$
follows from
\begin{equation}
 h_0+\sum_{i=0}^4k_i\al_i={T_0}^{k_0}\circ{T_1}^{k_1}\circ{T_2}^{k_2}\circ{T_3}^{k_3}\circ{T_4}^{k_4}(h_0),
\end{equation}
where $k_i\in\bbZ$.
Therefore, the case $i=0$ holds.
In a similar manner, we can also prove the statements for $M_i$ $i=1,\dots,4$. 
Therefore, we have completed the proof.
\end{proof}

Since of Lemma \ref{lemma:orbits}, we can express the lattices $M_i$ by
\begin{equation}
 M_i=\left\{T(h_i)\,\left|\, T\in \langle T_1,\dots,T_4\rangle\right\}\right.,\quad 
 i=0,\dots,4.
\end{equation}
In general, the translations on the orbits $M_i$, $i=0,\dots,4$, are given by the Kac translations defined on the weight lattice $P$:
\begin{equation}
 T_\al(\la)=\la-\langle\oc\de,\la\rangle\al,\quad
 \la\in P,\quad \al\in Q.
\end{equation}
The translations $T_i$, $i=0,\dots,4$, can be expressed by the Kac translations as the following:
\begin{equation}
 T_i=T_{\al_i},\quad i=0,\dots,4.
\end{equation}
Note that the translations $T_i$, $i=0,\dots,4$, are called the fundamental translations on the weight lattice $P$ or in affine Weyl group $W(A_4^{(1)})$.
Indeed, since of the following property of the Kac translations:
\begin{equation}
 T_\al\circ T_\beta=T_{\al+\beta},
\end{equation}
where $\al,\beta\in Q$,
all Kac translations in $W(A_4^{(1)})$ can be expressed by the compositions of $T_i$, $i=0,\dots,4$.

In this case we do not have translations moving a fundamental weight $h_i$ to another fundamental weight $h_j$. 
However, by extending $W(A_4^{(1)})$ with the automorphisms of the Dynkin diagram, we can define such translations as shown in the following subsection.

\subsection{Extended affine Weyl group of type $A_4^{(1)}$}
In this section, we consider the extended affine Weyl group of type $A_4^{(1)}$.

We define the transformations $\si$ and $\iota$ by
\begin{subequations}\label{eqn:dynkinauto_A4_weight}
\begin{align}
 &\si:h_0\to h_1,\quad
 h_1\to h_2,\quad
 h_2\to h_3,\quad
 h_3\to h_4,\quad
 h_4\to h_0,\\
 &\iota:h_1\leftrightarrow h_4,\quad h_2\leftrightarrow h_3,
\end{align}
\end{subequations}
whose actions on the root lattice are given by
\begin{subequations}\label{eqn:dynkinauto_A4_root}
\begin{align}
 &\si:\al_0\to \al_1,\quad
 \al_1\to \al_2,\quad
 \al_2\to \al_3,\quad
 \al_3\to \al_4,\quad
 \al_4\to \al_0,\\
 &\iota:\al_1\leftrightarrow \al_4,\quad \al_2\leftrightarrow \al_3.
\end{align}
\end{subequations}
Moreover, we also define their actions on the coroot lattice $\oc Q$ by replacing $\al_i$ with $\oc\al_i$ in the actions \eqref{eqn:dynkinauto_A4_root}.
Then, the transformations $\si$ and $\iota$ satisfy
\begin{equation}
 \sigma^5=\iota^2=1,\quad
 \sigma\circ\iota=\iota\circ\sigma^{-1},
\end{equation}
and the relations with $W(A_4^{(1)})=\langle s_0,s_1,s_2,s_3,s_4\rangle$ are given by
\begin{equation}
 \sigma\circ s_i=s_{i+1}\circ\sigma,\quad
 \iota\circ s_i=s_{-i}\circ\iota,
\end{equation}
that is, the transformations $\si$ and $\iota$ are automorphisms of the Dynkin diagram of type $A_4^{(1)}$ (see Figure \ref{fig:dynkin_A4}).
Therefore, we call 
\begin{equation}
 \tW(A_4^{(1)})=W(A_4^{(1)})\rtimes\langle \si,\iota\rangle
\end{equation}
as the extended affine Weyl group of type $A_4^{(1)}$.

\begin{remark}
We can easily verify that $\oc\de$ is also invariant under the action of $\tW(A_4^{(1)})$, 
and the transformations in $\tW(A_4^{(1)})$ preserve the form $\langle\cdot,\cdot\rangle$.
\end{remark}

Let $M$ be the orbits of $h_0$ defined by  
\begin{equation}\label{eqn:def_M0_extended}
 M=\left\{w(h_0)\,\left|\, w\in \tW(A_4^{(1)})\right\}\right..
\end{equation}
Moreover, we also define the following transformations:
\begin{subequations}\label{eqn:translation_Tij_A4}
\begin{align}
 &T_{01}=\si^4 s_{2340},\quad
 T_{12}=\si^4 s_{3401},\quad
 T_{23}=\si^4 s_{4012},\\
 &T_{34}=\si^4 s_{0123},\quad
 T_{40}=\si^4 s_{1234},
\end{align}
\end{subequations}
whose actions on the fundamental weights are translational as the following:
\begin{equation}
 T_{i\, i+1}(h_j)=h_j+v_{i\, i+1},\quad
 T_{i\, i+1}(v_{j\,j+1})=v_{j\,j+1},\quad
 i,j\in\bbZ/(5\bbZ),
\end{equation}
where
\begin{equation}
 v_{01}=h_0-h_1,~
 v_{12}=h_1-h_2,~
 v_{23}=h_2-h_3,~
 v_{34}=h_3-h_4,~
 v_{40}=h_4-h_0.
\end{equation}
Note that
\begin{align}
 &\begin{pmatrix}
 \al_0\\\al_1\\\al_2\\\al_3\\\al_4
 \end{pmatrix}
 =\left(\begin{array}{ccccc}
 1&0&0&0&-1\\
 -1&1&0&0&0\\
 0&-1&1&0&0\\
 0&0&-1&1&0\\
 0&0&0&-1&1
 \end{array}\right).
 \begin{pmatrix}
 v_{01}\\v_{12}\\v_{23}\\v_{34}\\v_{40}
 \end{pmatrix},\\
 & \begin{pmatrix}
 v_{01}\\v_{12}\\v_{23}\\v_{34}\\v_{40}
 \end{pmatrix}
 =\frac{1}{5}\left(\begin{array}{ccccc}
 4&0&1&2&3\\
 3&4&0&1&2\\
 2&3&4&0&1\\
 1&2&3&4&0\\
 0&1&2&3&4
 \end{array}\right).
 \begin{pmatrix}
 \al_0\\\al_1\\\al_2\\\al_3\\\al_4
 \end{pmatrix},\\
 &v_{01}+v_{12}+v_{23}+v_{34}+v_{40}=0,\\
 &T_{01}\circ T_{12}\circ T_{23}\circ T_{34}\circ T_{40}=1.
\end{align}
In a similar manner as the proof of Lemma \ref{lemma:orbits}, by using the translations $T_{ij}$ we can prove the following:
\begin{equation}
 M=\left\{h_0+v\,\left|\, v\in V \right\}\right.
 =\left\{T(h_0)\,\left|\, T\in \langle T_{01},T_{12},T_{23},T_{34},T_{40}\rangle\right\}\right.,
\end{equation}
where $V=\bbZ v_{01}+\bbZ v_{12}+\bbZ v_{23}+\bbZ v_{34}+\bbZ v_{40}$.
Note that since $T_{v_{i\, i+1}}:h_{i+1}\mapsto h_i$,
the following hold:
\begin{equation}
 h_i\in M,\quad i=0,\dots,4.
\end{equation}


The translations on the weight lattice $P$ spanned by $V$ is given by
\begin{equation}
 T_v(\la)=\la-\langle\oc\de,\la\rangle v,\quad
 \la\in P,\quad v\in V.
\end{equation}
In this case, the translations $T_{i\, i+1}$, $i\in\bbZ/(5\bbZ)$, can be expressed by
\begin{equation}
 T_{01}=T_{v_{01}},\quad
 T_{12}=T_{v_{12}},\quad
 T_{23}=T_{v_{23}},\quad
 T_{34}=T_{v_{34}},\quad
 T_{40}=T_{v_{40}},
\end{equation}
and are called the fundamental translations in $\tW(A_4^{(1)})$,
that is, all translations in $\tW(A_4^{(1)})$ can be expressed as the compositions of these translations.
Note that the fundamental translations in $W(A_4^{(1)})$ can be expressed as the compositions of the fundamental translations in $\tW(A_4^{(1)})$ as the following:
\begin{equation}
 T_i=T_{i\, i+1}\circ {T_{i-1,i}}^{-1},\quad i\in\bbZ/(5\bbZ).
\end{equation}

For any $v=\sum_{i=0}^4c_i\al_i\in V$, where $c_i\in \bbR$, we define the squared length of $T_v$ by
\begin{equation}
 |T_v|^2:=\langle \oc v,v\rangle
 =(c_0-c_1)^2+(c_1-c_2)^2+(c_2-c_3)^2+(c_3-c_4)^2+(c_4-c_0)^2,
\end{equation}
where $\oc v=\sum_{i=0}^4c_i \oc\al_i$.
Note that we here extended the domain of the bilinear pairing $\langle\cdot,\cdot\rangle$ from $\oc Q\times P$ to $\oc{\overline{Q}}\times P$, where
$\oc {\overline{Q}}=\sum_{k=0}^4\bbR\oc\al_k$.
We can easily verify that the squared length of fundamental translations in $W(A_4^{(1)})$: $T_i$ is $2$, while
the squared length of fundamental translations in $\tW(A_4^{(1)})$: $T_{i\, i+1}$ is $4/5$.
\section{General Elliptic Difference Equations}
\label{appendix:general_elliptic}
In this section, we provide two generic elliptic difference equations. The first is Sakai's $A_0^{(1)}$-surface equation. 
This was re-expressed by Murata \cite{mu:04} as  follows:
\begin{equation}
 T^{(M)}:(f,g;t,b_1,\dots,b_8)\mapsto\left(\overline{f},\overline{g};t+\frac{\delta}{2},b_1,\dots,b_8\right),
\end{equation}
where $\overline f$ and $\overline g$ are given by
\begin{subequations}
\begin{align}
 &\det\Big(v(f,g),v_1,\ldots, v_8,v_c\Big)
 \det\Big(v(\overline f, g),\widecheck{v}_1,\dots,\widecheck{v}_8,\widecheck{v}_c\Big)\notag\\
 &\qquad= P_+\,(f-f_c)(\overline f -\overline{f_c})\,\prod_{i=1}^8(g-g_i),\\[0.5em]
 &\det\Big(v(g,\overline f),\widehat{u}_1,\dots, \widehat{u}_8,\widehat{u}_c\Big)\det\Big(v(\overline g, \overline f),\overline{u}_1,\ldots, \overline{u}_8,\overline{u}_c\Big)\notag\\
 &\qquad=\overline{P}_-\,(g-g_c)(\overline g -\overline{g}_c)\,\prod_{i=1}^8(\overline f-\overline{f}_i).
\end{align}
\end{subequations}
Here,
\begin{equation}
 \delta=\sum_{k=1}^8 b_k,\quad
 v_i=v(f_i, g_i),\quad
 \widecheck{v}_i=v(\overline{f}_i, g_i),\quad 
 \widehat{u}_i=v(g_i, \overline{f}_i),\quad \overline{u}_i=v(\overline{g}_i,\overline{f}_i),
\end{equation}
for i=1,\dots,8,c, and
\begin{equation}
\begin{split}
 &f_i=\wp(t-b_i),\quad g_i=\wp(t+b_i),\quad i=1,\dots,8,\\
 &f_c=\wp\left(t+\frac{t^2}{\delta}\right),\quad
 g_c=\wp\left(t-\frac{t^2}{\delta}\right).
\end{split}
\end{equation}
Moreover, $v(a,b)$ and $P_\pm$ are given by
\begin{align*}
&\hspace*{-2em}v(a,b)=(ab^4, ab^3, ab^2,ab,a,b^4,b^3,b^2,b,1)^T,\\
&\hspace*{-2em}P_\pm=\frac{\sigma(4t)^4\sigma(4t\pm \delta)^4}{\sigma\left(t\mp \frac{t^2}{\delta}\right)^{16}}
\prod_{1\le i<j\le 8}\sigma(b_i-b_j)^2
\prod_{i=1}^8\frac{\sigma\left(\frac{t^2}{\delta}-b_i\right)\sigma\left(2 t\pm \frac{t^2}{\delta}\pm b_i\right)}{\sigma(t\pm b_i)^{14}\sigma(t\mp b_i)^{2}\sigma\left(t\mp b_i\pm\frac{\delta}{2}\right)^{2}},
\end{align*}
where $\sigma$ is the Weierstrass sigma function; see \cite[Chapter 23]{NIST:DLMF}.
Note that $\overline{P}_-=P_-|_{t\to\overline{t}}$.
The above system was obtained by deducing translations on the lattice of type $E_8^{(1)}$. 
We note that this translation corresponds to NVs in the lattice.

The second case of an elliptic difference equation was found by Joshi and Nakazono \cite{JN2017:MR3673467}. 
It has a projective reduction to the RCG equation \eqref{eqn:RCGeqn}. 
The generic equation is given by
\begin{align}
 T^{(JN)}:&(x,y;c_1,\dots,c_4,c_5,\dots,c_8,\eta)\notag\\
 &\mapsto(\ox,\oy;c_1-\la,\dots,c_4-\la,c_5+\la,\dots,c_8+\la,\eta+\la),
\end{align}
where $\ox$ and $\oy$ are given by
\begin{subequations}
\begin{align}
 &\bfrac{k\,\cd{\eta-\cc8+\ka}\oy+1}{k\,\cd{\eta-\cc7+\ka}\oy+1}
 \bfrac{\tx-\cd{\eta-\cc7+\frac{\cc{5678}}{2}+\la+\ka}}{\tx-\cd{\eta-\cc8+\frac{\cc{5678}}{2}+\la+\ka}}\notag\\
 &\qquad =\gG{\frac{\cc{5678}-2\cc5+\la}{2},\frac{\cc{5678}-2\cc6+\la}{2},\frac{\cc{5678}-2\cc7+\la}{2},\frac{\cc{5678}-2\cc8+\la}{2},\eta+\frac{\la}{2}+\ka}\notag\\
 &\qquad\qquad \frac{\pP{\frac{\cc{5678}-2\cc5+\la}{2},\frac{\cc{5678}-2\cc6+\la}{2},\frac{\cc{5678}-2\cc7+\la}{2},\eta+\frac{\la}{2}+\ka}{\tx,\ty}}
 {\pP{\frac{\cc{5678}-2\cc5+\la}{2},\frac{\cc{5678}-2\cc6+\la}{2},\frac{\cc{5678}-2\cc8+\la}{2},\eta+\frac{\la}{2}+\ka}{\tx,\ty}},\\
 &\bfrac{k\,\cd{\eta+\cc4+\ka}\ox+1}{k\,\cd{\eta+\cc3+\ka}\ox+1}
 \bfrac{k\,\cd{\eta-\cc3+2\la+\ka}\oy+1}{k\,\cd{\eta-\cc4+2\la+\ka}\oy+1}\notag\\
 &\qquad=\gG{\eta-\cc1+\frac{\cc{1234}}{4}+\la,\eta-\cc2+\frac{\cc{1234}}{4}+\la,\eta-\cc3+\frac{\cc{1234}}{4}+\la,\eta-\cc4+\frac{\cc{1234}}{4}+\la,\frac{\cc{5678}+2\la}{4}+\ka}\notag\\
 &\qquad\qquad \frac{\pP{\eta-\cc1+\frac{\cc{1234}}{4}+\la,\eta-\cc2+\frac{\cc{1234}}{4}+\la,\eta-\cc3+\frac{\cc{1234}}{4}+\la,\frac{\cc{5678}+2\la}{4}+\ka}{\frac{-1}{k\oy},\tx}}
 {\pP{\eta-\cc1+\frac{\cc{1234}}{4}+\la,\eta-\cc2+\frac{\cc{1234}}{4}+\la,\eta-\cc4+\frac{\cc{1234}}{4}+\la,\frac{\cc{5678}+2\la}{4}+\ka}{\frac{-1}{k\oy},\tx}},
\end{align}
and $\tx$ and $\ty$ are given by
\begin{align}
 &\bfrac{k\,\cd{\eta+\cc8-\frac{\cc{5678}}{2}}\ty+1}{k\,\cd{\eta+\cc7-\frac{\cc{5678}}{2}}\ty+1}
 \bfrac{x-\cd{\eta+\cc7}}{x-\cd{\eta+\cc8}}\notag\\
 &\qquad=\gG{\cc5,\cc6,\cc7,\cc8,\eta}
 \frac{\pP{\cc5,\cc6,\cc7,\eta}{x,y}}{\pP{\cc5,\cc6,\cc8,\eta}{x,y}},\\
 &\bfrac{k\,\cd{\eta-\cc4+\frac{\cc{1234}}{2}}\tx+1}{k\,\cd{\eta-\cc3+\frac{\cc{1234}}{2}}\tx+1}
 \bfrac{k\,\cd{\eta+\cc3+\frac{\cc{5678}}{2}}\ty+1}{k\,\cd{\eta+\cc4+\frac{\cc{5678}}{2}}\ty+1}\notag\\
 &\qquad=\gG{\eta+\cc1+\frac{\cc{5678}}{4},\eta+\cc2+\frac{\cc{5678}}{4},\eta+\cc3+\frac{\cc{5678}}{4},\eta+\cc4+\frac{\cc{5678}}{4},\frac{\cc{5678}}{4}}\notag\\
 &\qquad\qquad\frac{\pP{\eta+\cc1+\frac{\cc{5678}}{4},\eta+\cc2+\frac{\cc{5678}}{4},\eta+\cc3+\frac{\cc{5678}}{4},\frac{\cc{5678}}{4}}{\frac{-1}{k\ty},x}}
 {\pP{\eta+\cc1+\frac{\cc{5678}}{4},\eta+\cc2+\frac{\cc{5678}}{4},\eta+\cc4+\frac{\cc{5678}}{4},\frac{\cc{5678}}{4}}{\frac{-1}{k\ty},x}}.
\end{align}
\end{subequations}
Here, 
$\la=\sum_{k=1}^8\cc{k}$,
$\ka$ is defined by \eqref{eqn:kappa},
$\cc{j_1\cdots j_n}$ is the summation of the parameters $\cc{i}$ (see \eqref{notation:c}),
and the functions 
$\gG{\aaa1,\aaa2,\aaa3,\aaa4,b}$,
$\pQ{\aaa1,\aaa2,\aaa3,\aaa4,\aaa5,b}{X}$ and 
$\pP{\aaa1,\aaa2,\aaa3,b}{X,Y}$
are given by
{\allowdisplaybreaks
\begin{align}
 &\hspace{-2.5em}\gG{\aaa1,\aaa2,\aaa3,\aaa4,b}
 =\bcfrac{1-\frac{\cd{\aaa4+\frac{\aaa1+\aaa2}{2}}}{\cd{\aaa2+\frac{\aaa1+\aaa2}{2}}}}
 {1-\frac{\cd{\aaa3+\frac{\aaa1+\aaa2}{2}}}{\cd{\aaa2+\frac{\aaa1+\aaa2}{2}}}}
 \bcfrac{1-\frac{\cd{b-\aaa4}}{\cd{b-\aaa1}}}{1-\frac{\cd{b-\aaa3}}{\cd{b-\aaa1}}}\notag\\
 &\hspace{-2.5em}\hspace{6.8em}\bcfrac{1-\frac{\cd{b+\aaa4-\frac{\aaa1+\aaa2+\aaa3+\aaa4}{2}}}{\cd{b+\aaa2+\frac{\aaa1+\aaa2+\aaa3+\aaa4}{2}}}}
 {1-\frac{\cd{b+\aaa3-\frac{\aaa1+\aaa2+\aaa3+\aaa4}{2}}}{\cd{b+\aaa2+\frac{\aaa1+\aaa2+\aaa3+\aaa4}{2}}}}
 \bcfrac{1-\frac{\cd{\aaa3+\frac{\aaa1+\aaa2}{2}}}{\cd{2b+\aaa2-\frac{\aaa1+\aaa2}{2}}}}
 {1-\frac{\cd{\aaa4+\frac{\aaa1+\aaa2}{2}}}{\cd{2b+\aaa2-\frac{\aaa1+\aaa2}{2}}}},\\
 &\hspace{-2.5em}\pQ{\aaa1,\aaa2,\aaa3,\aaa4,\aaa5,b}{X}\notag\\
 &\hspace{-2.5em}=\left(\cd{b+\aaa3-\frac{\aaa5}{2}}-\cd{b+\aaa2+\frac{\aaa5}{2}}\right)
 \left(\cd{b+\aaa1+\frac{\aaa5}{2}}-\cd{b+\aaa4+\frac{\aaa5}{2}}\right)\notag\\
 &\hspace{-2.5em}\quad\Big(\cd{b+\aaa4}\cd{b+\aaa1}+\cd{b+\aaa2}X\Big)
 +\left(\cd{b+\aaa3-\frac{\aaa5}{2}}-\cd{b+\aaa1+\frac{\aaa5}{2}}\right)\notag\\
 &\hspace{-2.5em}\quad\left(\cd{b+\aaa4+\frac{\aaa5}{2}}-\cd{b+\aaa2+\frac{\aaa5}{2}}\right)
 \Big(\cd{b+\aaa4}\cd{b+\aaa2}+\cd{b+\aaa1}X\Big)\notag\\
 &\hspace{-2.5em}-\left(\cd{b+\aaa3-\frac{\aaa5}{2}}-\cd{b+\aaa4+\frac{\aaa5}{2}}\right)
 \left(\cd{b+\aaa1+\frac{\aaa5}{2}}-\cd{b+\aaa2+\frac{\aaa5}{2}}\right)\notag\\
 &\hspace{-2.5em}\quad\Big(\cd{b+\aaa1}\cd{b+\aaa2}+\cd{b+\aaa4}X\Big),\\
 &\hspace{-2.5em}\pP{\aaa1,\aaa2,\aaa3,b}{X,Y}=C_1 X Y+C_2 X+C_3 Y+C_4,
\end{align}
where
\begin{align*}
 &\hspace*{-2.5em}C_1=\Big(\cd{b-\aaa3}-\cd{b-\aaa2}\Big)\cd{b+\aaa1}
 +\Big(\cd{b-\aaa1}-\cd{b-\aaa3}\Big)\cd{b+\aaa2}\notag\\
 &+\Big(\cd{b-\aaa2}-\cd{b-\aaa1}\Big)\cd{b+\aaa3},\\
 &\hspace*{-2.5em}C_2=\Big(\cd{b-\aaa2}-\cd{b-\aaa3}\Big)\cd{b-\aaa1}\cd{b+\aaa1}\notag\\
 &+\Big(\cd{b-\aaa3}-\cd{b-\aaa1}\Big)\cd{b-\aaa2}\cd{b+\aaa2}\notag\\
 &+\Big(\cd{b-\aaa1}-\cd{b-\aaa2}\Big)\cd{b-\aaa3}\cd{b+\aaa3},\\
 &\hspace*{-2.5em}C_3=\Big(\cd{b+\aaa3}-\cd{b+\aaa2}\Big)\cd{b-\aaa1}\cd{b+\aaa1}\notag\\
 &+\Big(\cd{b+\aaa1}-\cd{b+\aaa3}\Big)\cd{b-\aaa2}\cd{b+\aaa2}\notag\\
 &+\Big(\cd{b+\aaa2}-\cd{b+\aaa1}\Big)\cd{b-\aaa3}\cd{b+\aaa3},\\
 &\hspace*{-2.5em}C_4=\Big(\cd{b+\aaa2}\cd{b-\aaa3}-\cd{b-\aaa2}\cd{b+\aaa3}\Big)\cd{b-\aaa1}\cd{b+\aaa1}\notag\\
 &+\Big(\cd{b+\aaa3}\cd{b-\aaa1}-\cd{b-\aaa3}\cd{b+\aaa1}\Big)\cd{b-\aaa2}\cd{b+\aaa2}\notag\\
 &+\Big(\cd{b+\aaa1}\cd{b-\aaa2}-\cd{b-\aaa1}\cd{b+\aaa2}\Big)\cd{b-\aaa3}\cd{b+\aaa3}.
\end{align*}
}
Similarly to the earlier case of Sakai, this equation also arises from a translation on the lattice of type $E_8^{(1)}$. 
We note that this translation corresponds to NNVs in the lattice.

\def\cprime{$'$} \def\cprime{$'$}

\end{document}